\documentclass[pra,twocolumn,superscriptaddress,showpacs]{revtex4-1}

\usepackage{amsmath}
\usepackage{latexsym}
\usepackage{amssymb}
\usepackage{graphicx}
\usepackage[colorlinks=true, citecolor=blue, urlcolor=blue]{hyperref}
\usepackage{float}
\usepackage{amsfonts}
\usepackage{textcomp}
\usepackage{mathpazo}
\usepackage{comment}

\usepackage{bbm}

\usepackage{xcolor}
\definecolor{myurlcolor}{rgb}{0,0,0.4}
\definecolor{mycitecolor}{rgb}{0,0.5,0}
\definecolor{myrefcolor}{rgb}{0.5,0,0}
\usepackage{hyperref}
\hypersetup{colorlinks,
linkcolor=myrefcolor,
citecolor=mycitecolor,
urlcolor=myurlcolor}

\usepackage[draft]{fixme}
\usepackage{amsmath,bbm}
\usepackage{graphicx}
\usepackage{amsfonts}
\usepackage{amssymb}
\usepackage{amsmath, amssymb, amsthm,verbatim,graphicx,bbm}
\usepackage{mathrsfs}
\usepackage{color,xcolor,longtable}

\newcommand{\beq}[0]{\begin{equation}}
\newcommand{\eeq}[0]{\end{equation}}

\newcommand{\one}{\leavevmode\hbox{\small1\normalsize\kern-.33em1}}

\font\Bbb =msbm10 
 \def\C{{\hbox {\Bbb C}}}

\def\be{\begin{equation}}
\def\ee{\end{equation}}
\def\ben{\begin{eqnarray}}
\def\een{\end{eqnarray}}
\def\eea{\end{array}}
\def\bea{\begin{array}}
\newcommand{\ot}[0]{\otimes}
\newcommand{\Tr}[1]{\mathrm{Tr}#1}
\newcommand{\bei}{\begin{itemize}}
\newcommand{\eei}{\end{itemize}}
\newcommand{\ket}[1]{|#1\rangle}
\newcommand{\bra}[1]{\langle#1|}

\newcommand{\proj}[1]{\ket{#1}\!\bra{#1}}

\newcommand{\N}{\mathbb{N}}
\newcommand{\Q}{\mathbb{Q}}
\newcommand{\R}{\mathbb{R}}
\newcommand{\I}{\mathbbm{1}}
\newcommand{\w}{\omega}
\newcommand{\p}{\Vec{p}}

\renewcommand{\emph}[1]{\textbf{#1}}

\makeatletter
\newtheorem*{rep@theorem}{\rep@title}
\newcommand{\newreptheorem}[2]{%
\newenvironment{rep#1}[1]{%
 \def\rep@title{#2 \ref{##1}}%
 \begin{rep@theorem}}%
 {\end{rep@theorem}}}
\makeatother

\theoremstyle{plain}
\newtheorem*{thm*}{Theorem}
\newreptheorem{thm}{Theorem}
\newtheorem{lem}{Lemma}
\newtheorem{fakt}{Fact}
\newtheorem{obs}{Observation}[lem]

\theoremstyle{definition}

\theoremstyle{remark}

\usepackage[T1]{fontenc}

\begin{document}
\title{Self-testing quantum systems of arbitrary local dimension\\ with minimal number of measurements}
\author{Shubhayan Sarkar}
\affiliation{Center for Theoretical Physics, Polish Academy of Sciences, Aleja Lotnikow 32/46, 02-668 Warsaw, Poland }
\author{Debashis Saha}
\affiliation{Center for Theoretical Physics, Polish Academy of Sciences, Aleja Lotnikow 32/46, 02-668 Warsaw, Poland }
\author{J\k{e}drzej Kaniewski}
\affiliation{Faculty of Physics, University of Warsaw, Pasteura 5, 02-093 Warsaw, Poland}
\author{Remigiusz Augusiak}
\affiliation{Center for Theoretical Physics, Polish Academy of Sciences, Aleja Lotnikow 32/46, 02-668 Warsaw, Poland }

\begin{abstract}
Bell nonlocality as a resource for device independent certification schemes has been studied extensively in recent years. The strongest form of device independent certification is referred to as self-testing, which given a device certifies the promised quantum state as well as quantum measurements performed on it without any knowledge of the internal workings of the device. 
In spite of various results on self-testing protocols, it remains a highly nontrivial problem to propose a certification scheme of qudit-qudit entangled states based on violation of a single $d$-outcome Bell inequality. Here we address this problem and propose a self-testing protocol for the maximally entangled state of any local dimension using the minimum number of measurements possible, i.e.,  two per subsystem.
Our self-testing result can be used to establish unbounded randomness expansion, 
$\log_2d$ perfect random bits, while it requires only one random bit to encode the measurement choice.
\end{abstract}

\maketitle

\textit{Introduction.---}
The advent of quantum theory has not just changed the understanding of physics but has also given rise to new phenomena that would have never been possible  in the classical world. Arguably one of the most interesting features of quantum theory is the existence of quantum correlations which cannot be explained by any local hidden variable model, a phenomenon commonly referred to as Bell nonlocality \cite{Bell,Bell66}. It has been understood that apart from its fundamental significance, nonlocality is a powerful resource for certain device-independent applications such as quantum cryptography \cite{di2}, randomness generation \cite{di3,di4}, or, more recently, for  device-independent certification methods \cite{di1,di2,di3,di4}.

The strength of the device-independent (DI) certification methods lies in the fact that given an unknown quantum system, one can make nontrivial statements on some of its key features based solely on the nonlocal correlations obtained from it.
They are advantageous over standard certification methods such as those based on quantum tomography (cf. Ref. \cite{NielsenChuang}) because they do not require making assumptions on the system under study, apart from that it is governed by quantum mechanics. An example of such a DI certification scheme would be verifying whether a quantum device produces entanglement \cite{DIEW} or certification of the dimension of a quantum system \cite{Dimension}, both based on a violation of some Bell inequality by the corresponding quantum system.

The strongest form of DI certification is self-testing. First introduced in \cite{di1}, it allows one to provide a full description, up to certain well-understood equivalences, of the considered quantum system and also the measurements performed on it based on observing the maximal violation of some Bell inequality. Such a form of certification is particularly interesting from the application point of view as it provides a way of verifying that a given quantum device functions properly without the need of knowing its internal working. 

A lot of attention has thus been devoted to proposing self-testing schemes for entangled quantum states  \cite{Scarani,Yang,Bamps,All,chainedBell}. However, most of the obtained results focus on states that are locally qubits such as for instance the self-testing statement for any two-qubit entangled state \cite{Yang,Bamps} based on a violation of the tilted version of the famous Clauser-Horne-Shimony-Holt (CHSH) Bell inequality \cite{CHSH,Acin}. At the same time, quantum systems of higher local dimension remain mostly unexplored and, in fact, few results are devoted to them \cite{Yang,Coladangelo,coladangelo18,Jed}. In particular, in Refs. \cite{Coladangelo} the two-qubit results \cite{Yang,Bamps} were combined to design a self-testing protocol for any entangled state of arbitrary local dimension. Still, these results are based on a violation of many two-outcome Bell inequalities and the question whether one can design a self-testing statement for qudit quantum systems relying on violation of a single and truly $d$-outcome Bell inequality remains open. Moreover, the self-testing statement of Refs. \cite{Yang,Coladangelo} is not optimal in terms of the number of measurements that the observers need to perform (three and four, respectively) to certify the state and the corresponding measurements. Taking into account the possibility of experimental implementations of self-testing, it is a highly relevant question whether alternative protocols can be derived which rely on the minimal number of two measurements per observer.

The main aim of this work is to address the above questions. We provide the first self-testing statement for a maximally entangled state of local dimension $d$:
\begin{equation}\label{phi+}
    \ket{\phi_d^+}=\frac{1}{\sqrt{d}}\sum_{i=0}^{d-1}\ket{ii}
\end{equation}
from a violation of a truly $d$-outcome Bell inequality exploiting the minimal possible number of measurements per party, which is two. To this end, we use the Bell inequality introduced in Ref. \cite{SATWAP} as a generalization of the well-known CHSH Bell inequality to $d$-outcome Bell scenarios. A straightforward implication of our self-testing statement is a novel and simpler, as compared to the previous approaches \cite{wu16,mckague16}, scheme for parallel self-testing of $N$ copies of the two-qubit maximally entangled state $\ket{\phi_2^+}$.
%
%
Another implication is that the outcomes of local measurements maximally violating the respective Bell inequality are perfectly random, which allows us to propose a quantum protocol for unbounded expansion of quantum randomness.

\textit{Preliminaries.---}Before presenting our results let us 
set up the scenario and introduce the relevant notions.

\textit{Bell scenario.} The simplest experimental set-up exhibiting quantum nonlocality, namely, the bipartite Bell set-up, comprises one preparation and two measurement devices. The latter are possessed by distant and noncommunicating parties, Alice and Bob, and can in general perform one of $m$ measurements, denoted $M_x$ and $N_y$ with $x,y=1,\ldots,m$. In each run of the experiment, a bipartite quantum system $\rho_{AB}$ is prepared by the preparation device and subsequently each measurement device performs a measurement on one of its subsystems, yielding an outcome $a$ and $b$, respectively, with $a,b\in\{0,\ldots,d-1\}$. 

Correlations obtained in this experiment are captured by a vector of probability distributions 
\begin{equation}
    \p = \{p(a,b|x,y)\}\in \R^{(md)^2}.
\end{equation}
Here, $p(a,b|x,y)$ is the probability of obtaining outcomes $a$
and $b$ by Alice and Bob after performing measurements $M_x$ and $N_y$, respectively, and it is given by the well-known formula
\begin{equation}
    p(a,b|x,y)=\Tr\left[\rho_{AB}(M_x^{(a)}\otimes N_y^{(b)})\right],
\end{equation}
where $M_x^{(a)}$ and $N_y^{(b)}$ are the measurement operators
defining the measurements $M_x$ and $N_y$, respectively. In what follows 
we also refer to the vector $\vec{p}$ as correlations.

In a fixed Bell scenario with $m$ $d$-outcome measurements such quantum correlations, that is, those obtained by performing local measurements on composite quantum systems, form a convex set $\mathcal{Q}$. For further purposes it is important to recall that any point $\p \in\mathcal{Q}$ can be achieved with a pure state $\ket{\psi_{AB}}$, whose local dimensions might in general be higher than those of $\rho_{AB}$, and projective measurements $M_x=\{P_x^{(a)}\}$ and $N_y=\{Q_y^{(b)}\}$. 

It turns out that the quantum set $\mathcal{Q}$ contains correlations that even if obtained from a quantum state, can be simulated by Alice and Bob in a purely classical way. Such correlations are said to admit a local-realistic description and for brevity we refer to them as local or classical. More formally, local correlations are those for which $p(a,b|x,y)$ can be represented as a convex combination of product deterministic correlations
\begin{equation}
p(a,b|x,y)=\sum_{\lambda}\mu(\lambda) p(a|x,\lambda)p(b|y,\lambda),
\end{equation}
where $\lambda$ is a random variable distributed according to a distribution $\mu(\lambda)$ and $p(a|x,\lambda),p(b|y,\lambda)\in\{0,1\}$ for every $x,y,a,b$. For any $m$ and $d$ such local correlations form a convex polytope that we denote $\mathcal{L}$.

\textit{Bell inequalities.} As proven by Bell, the local set $\mathcal{L}$ is a proper subset of $\mathcal{Q}$ \cite{Bell} and those quantum points that are outside $\mathcal{L}$ are termed nonlocal. The most natural way to show that a given point $\p\in\mathcal{Q}$ is not an element of $\mathcal{L}$, is to use Bell inequalities. Recall their generic form to be 
\begin{eqnarray}\label{BI}
\mathcal{I}: = \Vec{t}\cdot\p &\leq& \beta_L, 
\end{eqnarray}
where $\vec{t}=\{t_{abxy}\}$ is a collection of some real coefficients $t_{abxy}$ and $\beta_L=\max_{\p\in\mathcal{L}} \mathcal{I}$ is the local bound of the inequality (\ref{BI}). Analogously, the quantum or the Tsirelson bound of (\ref{BI}) is defined as $\beta_Q=\sup_{\p\in\mathcal{Q}}\mathcal{I}$.

Clearly, violation of a Bell inequality by correlations $\vec{p}$ implies that it is nonlocal. Moreover, any point $\vec{p}$ violating maximally some Bell inequality (or, equivalently, achieving its Tsirelson bound) necessarily belongs to the boundary of the quantum set $\mathcal{Q}$.

\textit{Correlation picture.} Let us finally notice that it is often beneficial to describe correlations obtained in the Bell experiment by expectation values instead of probability distributions. A convenient way to do so in Bell scenarios involving $d$-outcome measurements
is to consider the two-dimensional Fourier transform of the conditional probabilities $p(a,b|x,y)$:
\begin{equation}\label{ExpValues}
    \langle A^{(k)}_x B^{(l)}_y \rangle = \sum^{d-1}_{a,b=0} \w^{ak+bl} p(a,b|x,y),
\end{equation}
where $\w$ is the $d$-th root of unity $\omega=\exp(2\pi\mathbbm{i}/d)$ and $k,l=0,\ldots,d-1$. Importantly, 
for any quantum point $\vec{p}$, the expectation values (\ref{ExpValues}) can be represented as $\langle A^k_x B^l_y \rangle = \langle\psi_{AB}|A_x^{(k)}\otimes B_y^{(l)}|\psi_{AB}\rangle$ with $\{A_x^{(k)}\}$ and $\{B_y^{(l)}\}$ being collections of unitary operators with eigenvalues $\omega^i$ $(i=0,\ldots,d-1)$ defined as the Fourier transforms of the corresponding projective measurements
\begin{equation}
    A^{(k)}_x = \sum^{d-1}_{a=0} \w^{ak} P^{(a)}_x, \qquad B^{(l)}_y = \sum^{d-1}_{b=0} \w^{bl} Q^{(b)}_y.
\end{equation}
It is not difficult to see that $A_x^{(k)}$ is simply the $k$th power of $A_x$
(and similarly for Bob's operators) and thus in what follows we simply write $A_x^k$. In this correlator picture the unitary operators $A_x$ and $B_y$ are $d$-outcome observables measured in the Bell set-up.

\textit{Self-testing.---}We are now ready to present our main result: the self-testing statement for the two-qudit maximally entangled state and certain $d$-outcome measurements usually referred to as Collins-Gisin-Linden-Massar-Popescu (CGLMP) measurements \cite{Zukowski,CGLMP,BKP} (see Appendix \ref{sec:cglmp} for their explicit form).

To recall the task of self-testing, or more generally, DI certification, let us consider a quantum device performing a Bell experiment on some state $\ket{\psi_{AB}}\in\mathcal{H}_A\otimes\mathcal{H}_B$ with some quantum $d$-outcome observables $A_x$ and $B_y$, where the dimensions of the $\mathcal{H}_A$ and $\mathcal{H}_B$ are unknown but finite. The only information accessible to the user about how this device functions are the observed correlations $\vec{p}$. The aim of DI certification is to reveal the form of the state and observables from violation of some Bell inequality by $\vec{p}$ (given that the observed correlations violate some Bell inequality).
However, there are two degrees of freedom which can never be detected from the observed statistics. One is the set of local unitaries $U_A,U_B$ that act on $\mathcal{H}_A,\mathcal{H}_B$, that is, the state $U_A\otimes U_B|\psi_{AB}\rangle$ together with $\{U_AA_xU^\dagger_A\},\{U_BB_yU^\dagger_B\}$ will generate the same $\p$. Another one is the presence of an auxiliary system on which the measurements act trivially. A particular instance of such DI certification, termed self-testing, infers a unique state and two sets of unique measurements up to these equivalences. 

Clearly, a necessary condition to derive an exact self-testing statement
for a state and measurements is that the obtained correlations $\p$ lie on the boundary of $\mathcal{Q}$, and thus violate some Bell inequality maximally. Consequently, in order to derive a self-testing statement for $\ket{\phi_d^+}$ for any $d$ the first task is to identify a class of Bell inequalities in the bipartite scenario with the minimum number of measurements for which $\beta_Q$ is achieved by it for any $d$. The only known Bell inequality meeting these requirements is the Salavrakos-Augusiak-Tura-Wittek-Ac\'in-Pironio (SATWAP) Bell inequality \cite{SATWAP}, which in the correlator picture reads
 \begin{eqnarray}\label{eqSATWAP}
\mathcal{I}_{d}:&=&\sum_{k=1}^{d-1} \Big( a_k \langle A_1^{k}B_1^{d-k}\rangle +a_k^*\omega^k \langle A_1^{k}B_2^{d-k} \rangle \nonumber\\
&& \;\;\qquad + a_k^* \langle A_2^{k} B_1^{d-k} \rangle +  a_k \langle A_2^{k} B_2^{d-k} \rangle \Big) \leq \beta_L,
\end{eqnarray}
where $a_k=(1+\mathbbm{i})\omega^{k/4}/2$ and the classical value is given by 
$\beta_L=(1/2)[3\cot(\pi/4d)-\cot(3\pi/4d)]-2$. 

As proven in Ref.~\cite{SATWAP}, the maximal quantum value of this inequality is $\beta_Q=2(d-1)$ and it is achieved by the maximally entangled state of two qudits $\ket{\phi_d^+}$ and the aforementioned optimal CGLMP measurements 
(cf. Appendix \ref{sec:cglmp}). In what follows we show that this is in fact the only quantum system, up to additional of freedom and local unitary operations, realizing the maximal quantum violation of the inequality (\ref{eqSATWAP}). 

Let us introduce $Z_d=\mathrm{diag}[1,\omega,\ldots,\omega^{d-1}]$ to be the
$d$-dimensional generalization of the $\sigma_z$-Pauli matrix in the standard basis and $T_d$ to be the following matrix
\begin{equation}\label{ZdTd1}
T_d=\sum_{i=0}^{d-1}\omega^{i+\frac{1}{2}}\proj{i}-\frac{2}{d}\sum_{i,j=0}^{d-1}(-1)^{\delta_{i,0}+\delta_{j,0}}\omega^{\frac{i+j+1}{2}}|i\rangle\!\langle j|,
\end{equation}
where $\delta_{i,j}$ is the Kronecker delta: $\delta_{i,j}=1$ for $i=j$
and $\delta_{i,j}=0$ otherwise. It is not difficult to see that both matrices are unitary and have eigenvalues $\omega^i$ $(i=0,\ldots,d-1)$ (see Appendix B)
and thus represent $d$-outcome  quantum observables.\\
\begin{figure}[h]
\centering
\includegraphics[width=0.5\textwidth]{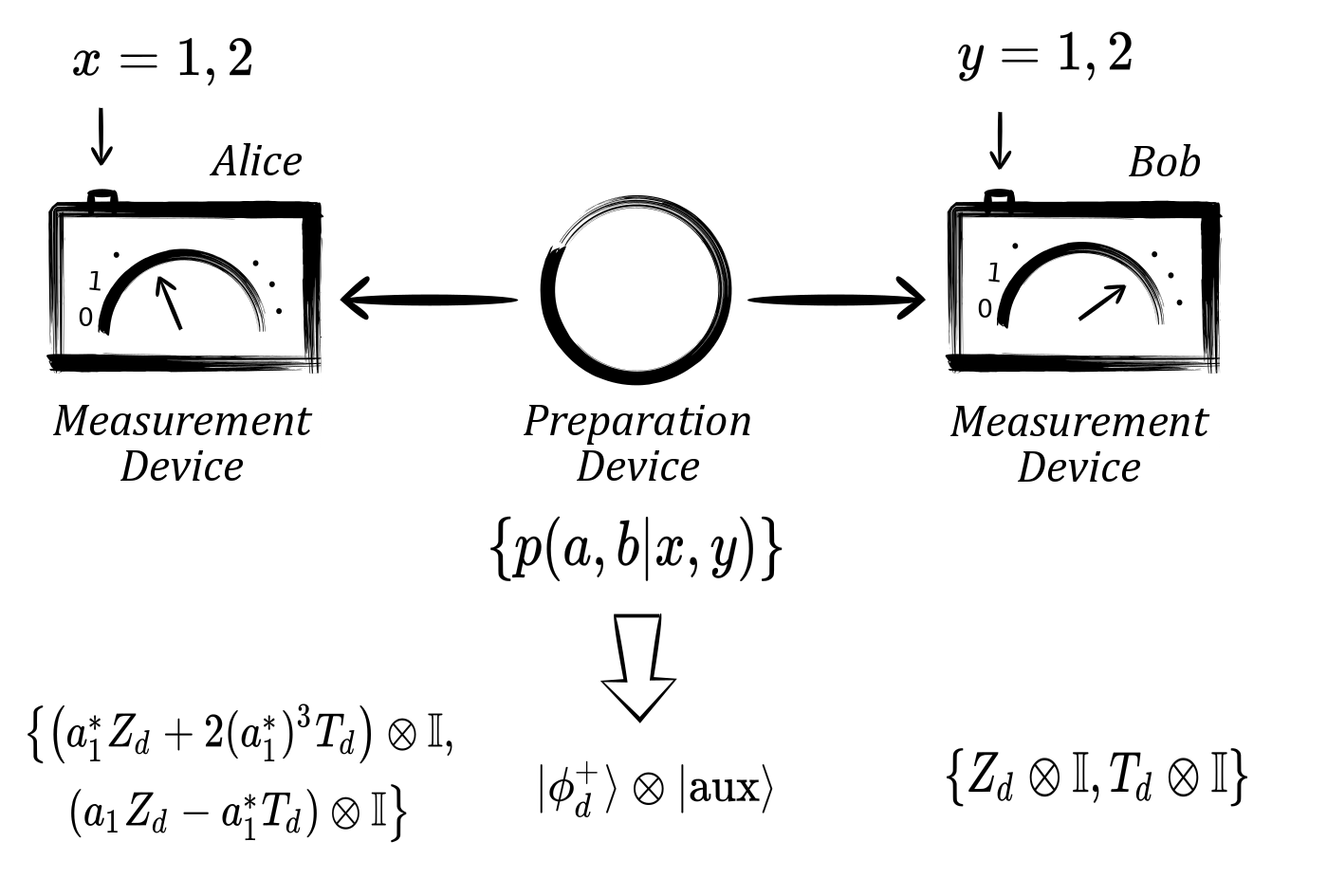}
\caption{\textbf{The self-testing scenario.} Alice and Bob receive an unknown quantum system from the preparation device. They have two input buttons each denoted by $1,2$ specifying two different measurements which they choose at random. The measurements point at one, out of the $d$ outcomes denoted by $\{0,1,\ldots,d-1\}$. Alice and Bob repeat the experiment to collect statistics. They compare their statistics and find the joint probability distribution for different outcomes and measurements denoted by $p(a,b|x,y)$. From the statistics we self-test the maximally entangled state \eqref{phi+} and two pairs of $d$-outcome measurements \eqref{AiBimain}-\eqref{AiBimain2}.}
\label{fig}
\end{figure}

Now, we can state our main theorem.
\begin{thm*}
Assume that the SATWAP Bell inequality (\ref{eqSATWAP}) is maximally violated by a state $\ket{\psi_{AB}}\in\mathcal{H}_A\otimes\mathcal{H}_B$ and unitary observables $A_i,B_i \ (i\in \{1,2\})$. Then, for any $d$, $\mathcal{H}_A=\C^d \otimes \mathcal{H}_{A'}$ as well as $\mathcal{H}_B=\C^d \otimes \mathcal{H}_{B'}$ with some auxiliary
Hilbert spaces $\mathcal{H}_{A'}$ and $\mathcal{H}_{B'}$ of unknown but finite dimensions. Moreover, there exist local unitary transformations
%
 $   U_A:\C^d \otimes \mathcal{H}_{A'} \rightarrow \C^d \otimes \mathcal{H}_{A'}$
 and $   U_B:\C^d \otimes \mathcal{H}_{B'} \rightarrow \C^d \otimes \mathcal{H}_{B'}$,
%
such that 
\begin{equation}\label{AiBimain}
U_B B_1 U^{\dagger}_B = Z_d \otimes \I_{B'}, \qquad U_B B_2 U^{\dagger}_B = T_d \otimes \I_{B'},    
\end{equation} 
and
\ben\label{AiBimain2}
U_A A_1 U^{\dagger}_A &=&   \left( a_1^* Z_d + 2(a^*_1)^3 T_d \right) \otimes \I_{A'}, \nonumber \\
U_A A_2 U^{\dagger}_A &=&   \left( a_1 Z_d - a^*_1 T_d \right) \otimes \I_{A'},
\een 
where $a_1 = (1+\mathbbm{i})\w^{1/4}/2$, $Z_d$ and $T_d$ are defined above,
and $\I_{A'},\I_{B'}$ are identity matrices acting on $\mathcal{H}_{A'}$ and $\mathcal{H}_{B'}$. Finally, the state $\ket{\psi_{AB}}$ is equivalent to 
$\ket{\phi_d^+}$ in the following sense
\be\label{state}
U_A\otimes U_B \ket{\psi_{AB}} = \ket{\phi^+_d} \otimes \ket{\mathrm{aux}_{A'B'}}, \ee 
where $\ket{\mathrm{aux}_{A'B'}}$ is some auxiliary state from 
$ \mathcal{H}_{A'}\otimes \mathcal{H}_{B'}$.
\end{thm*}
Before proving this theorem a few comments are in order. First, as proven in Appendix \ref{sec:cglmp}, the observables $Z_d$ and $T_d$ are unitarily equivalent to 
the CGLMP measurements. Second, our self-testing statement exploits a single Bell inequality involving the minimal number of two measurements per observer for which one can observe nonlocality in quantum systems. The number of local observables one needs to measure in order to self-test a state is of key importance from the application point of view.
Finally, in the case of $d=2^N$, the theorem implies a  parallel self-testing protocol for $\ket{\phi^+_2}^{\otimes N}$, i.e., $N$ copies of maximally two-qubit entangled state (see Refs. \cite{wu16,mckague16} for the previous approaches).
\begin{proof}
The proof is quite technical and thus is deferred to 
Appendix \ref{sec:selftesting}. Here we only sketch its main ingredients.

Let a state $\ket{\psi_{AB}}$ and observables $A_x$ and $B_y$
maximally violate the SATWAP Bell inequality. Without loss of generality, we can assume that the reduced states of $\ket{\psi_{AB}}$ are full-rank. Exploiting the sum-of-squares decompositions of the corresponding Bell operator we can establish that 
Alice's and Bob's observables satisfy
%
\begin{equation}
    \Tr(A_x^n)=\Tr(B_y^n)=0
\end{equation}
for $x,y=1,2$ and any $n<d$ which is a divisor of the number of outcomes $d$. These conditions imply that the multiplicities of the eigenvalues of
all the observables $A_x$ and $B_y$ are equal, which has two consequences. First, their matrix dimensions are a multiple of $d$, meaning that they act on, respectively, $\mathcal{H}_A=\mathbb{C}^d\otimes\mathcal{H}_{A'}$ and $\mathcal{H}_B=\mathbb{C}^d\otimes\mathcal{H}_{B'}$ with $\mathcal{H}_{A'}$ and $\mathcal{H}_{B'}$ being some auxiliary Hilbert spaces of unknown but finite dimensions. Second, there exist unitary operations
$U_A:\mathcal{H}_A\to \mathcal{H}_{A}$ and $U_B:\mathcal{H}_B\to \mathcal{H}_{B}$ which allow us to bring Alice's and Bob's observables
to the form given in Eqs. (\ref{AiBimain}) and (\ref{AiBimain2}).

Finally, to obtain an analogous statement for the state we again employ the sum-of-squares decomposition, which in virtue of the relations (\ref{AiBimain}) and (\ref{AiBimain2}), imposes the following conditions on the state $\ket{\psi}_{AB}$:
\begin{eqnarray}
&    [(Z_d^*)^k\otimes Z^k_d\otimes\mathbbm{1}_{A'B'}]\ket{\hat{\psi}_{AB}}=\ket{\hat{\psi}_{AB}},&\nonumber\\
&    [(T_d^*)^k\otimes T^k_d\otimes\mathbbm{1}_{A'B'}]\ket{\hat{\psi}_{AB}}=\ket{\hat{\psi}_{AB}},&
\end{eqnarray}
where we have denoted $\ket{\hat{\psi}_{AB}}=U_A\otimes U_B\ket{\psi_{AB}}$. These conditions can easily be solved leading to Eq. (\ref{state}), which completes the proof.
\end{proof}

{\it Applications to randomness expansion.---}
An interesting application of our self-testing statement is DI randomness certification \cite{di4}. Let us consider the distrustful scenario where one of the observers performing the Bell test, say Bob, wants to generate random bits from the outcomes of his measurements, and Eve who supplies the measurement devices wants to access that outcomes. 

We show with the aid of the above self-testing result that the maximal violation of the SATWAP Bell inequality certifies $\log_2 d$ bits of perfect 
randomness in the outcomes of Alice's and Bob's measurements. Without loss of generality we can focus on Bob's measurements. One quantifies the randomness obtained in a Bell experiment as $-\log_2 G(x,\vec{p})$, where $G(x,\vec{p})$ is the local guessing probability defined through the following formula
\begin{eqnarray}
G(y,\p) = \sup_{S\in S_{\p}} \sum_b\bra{\psi_{ABE}}\I\otimes Q_y^{(b)}\otimes
E^{(b)}\ket{\psi_{ABE}},
\end{eqnarray}
where $\ket{\psi_{ABE}}$ is a three-partite state shared by Alice, Bob and Eve, 
$\{Q_{y}^{(b)}\}$ is a projective measurement performed by Bob, and $E=\{E^{(c)}\}$ is a $d$-outcome measurement on $\mathcal{H}_E$ whose result is Eve's best guess of Bob's outcome. Finally, $S_{\p}$ is the set of all possible strategies that Eve implements to guess the outcome of Bob's measurements, consisting of the state $\ket{\psi_{ABE}}$ and the measurement $E$, which reproduce the correlations $\vec{p}$ observed by Alice and Bob, i.e., 
\begin{equation}
p(a,b|x,y)=\langle\psi_{ABE}|P_x^{(a)}\otimes Q_{y}^{(b)}\otimes \I_E|\psi_{ABE}\rangle.
\end{equation}
Assume now that the inequality (\ref{eqSATWAP}) is maximally violated
by $\vec{p}$. This means that, up to local unitary operations, $\ket{\psi_{ABE}}=\ket{\phi_{d}^+}\ket{\mathrm{aux}_{A'B'E}}$ as well as
$P_x^{(a)}=\bar{P}_x^{(a)}\otimes\I_{A'}$ and $Q_y^{(b)}=\bar{Q}_y^{(b)}\otimes\I_{B'}$, where $\bar{P}_x^{(a)}$
and $\bar{Q}_y^{(b)}$ are eigenprojectors of the observables in Eqs. 
(\ref{AiBimain}) and (\ref{AiBimain2}), respectively. Taking all this into account along with the fact that Bob's observables are traceless, one finally finds that $G(y,\vec{p})=1/d$.

Thus,  $-\log_2 G(x,\p)=\log_2d$ bits of randomness can be certified using our self-testing statement. This is the maximum randomness that could be extracted from a system of local dimension $d$, whereas it requires one bit of randomness to encode the inputs. This gives rise to unbounded randomness expansion as $d$ can be arbitrary large.

{\it Conclusion and discussion.---}We propose the first, to the best of our knowledge self-testing statement for quantum system of arbitrary local dimension that unlike the previous approaches \cite{Yang,Coladangelo} does not rely on self-testing results for qubit states and exploits a truly $d$-outcome Bell inequality. Moreover, our self-testing result exploits only two measurements per party, which is the minimal number allowing 
the parties to observe nonlocal correlations.  This makes our results interesting from the experimental point of view. In fact, violation of the SATWAP Bell inequality has already been experimentally tested in Ref. \cite{Science}.


Several follow-up questions arise from our work. First, it would be interesting to study robustness of our self-testing statement and compare it to that of the scheme of \cite{Yang,Coladangelo}. Another interesting direction is to explore whether the $d$-outcome SATWAP Bell inequality can be modified to be maximally violated by partially entangled state and whether the resulting modifications can be used to make self-testing statements for those states, again with the minimal number of measurements per observer. Finally, it is also interesting to investigate whether, analogously to Refs. \cite{Optimal,randomness2}, our self-testing statement can be used to establish certification of the optimal amount of $2\log_2 d$ bits of local or global randomness from the maximally entangled state of two qudits with the aid of a local POVM or two projective measurements, respectively.
From a more general perspective, our self-testing result might provide non-trivial insights into the structure of the set of quantum correlations along the lines of \cite{andrea}, or might find applications in delegated quantum computation \cite{Delegated}. 


\textit{Note added.---}While working on this project we became aware of the work \cite{related}
exploring parallel self-testing with minimal number of measurements.

\textit{Acknowledgement.---}We would like to thank A. Ac\'in, P. Nurowski and A. Sawicki for useful discussions. This work is supported by Foundation for Polish Science through the First Team project
(No First TEAM/2017-4/31). JK acknowledges support from the Homing programme of the Foundation for Polish Science (grant agreement no.~POIR.04.04.00-00-5F4F/18-00) co-financed by the European Union under the European Regional Development Fund.

\clearpage

\onecolumngrid

\appendix

\section{Proof of self-testing}\label{sec:selftesting}

In this appendix we provide the full proof of theorem stated in the main text.
However, before diving into more detailed considerations we recall certain facts that will be used throughout the proof. First, let us recall that the explicit form of 
the SATWAP Bell inequality with two measurements per observer is
 \begin{equation}\label{app:SATWAP}
\mathcal{I}_{d}=\sum_{k=1}^{d-1} \left( a_k \langle A_1^{k}B_1^{d-k}\rangle +a_k^*\omega^k \langle A_1^{k}B_2^{d-k} \rangle 
+ a_k^* \langle A_2^{k} B_1^{d-k} \rangle +  a_k \langle A_2^{k} B_2^{d-k} \rangle \right) \leq \beta_{C}^{d},
\end{equation}
where 
\begin{equation}\label{app:ak}
     a_k=\frac{1}{\sqrt{2}}\omega^{\frac{2k-d}{8}}=\frac{1-\mathbbm{i}}{2}\omega^{k/4}=\frac{1-\mathbbm{i}}{2}\mathrm{e}^{\frac{\pi \mathbbm{i}k}{2d}}
\end{equation}
and the classical bound $\beta_{C}^d$ is given explicitly by 
\begin{equation}
    \beta_{C}^d=\frac{1}{2}\left[3\cot\left(\frac{\pi}{4d}\right)
    -\cot\left(3\frac{\pi}{4d}\right)\right]-2.
\end{equation}
Then, the maximal quantum value of $\mathcal{I}_d$ turns out to be $\beta_Q^d=2(d-1)$. 
Let us also recall that the sum-of-squares decomposition of the corresponding Bell operator $\mathcal{B}_d$ is
\begin{equation}
    \beta_{Q}^d\mathbbm{1}-\mathcal{B}_d=\frac{1}{2}\sum_{k=1}^{d-1}\left(P_{1,k}^{\dagger}P_{1,k}+P_{2,k}^{\dagger}P_{2,k}\right),
\end{equation}
with
\begin{equation}\label{SOSR1}
    P_{i,k}=\mathbbm{1}-A_{i}^{k}\otimes C_{i}^{(k)}
\end{equation}
for $i=1,2$ and $k=1,\ldots,d-1$, where
\begin{equation}\label{Coperators}
C_1^{(k)}=a_k B_1^{-k} + a_{k}^{*}\omega^k B_2^{-k}, \qquad 
C_2^{(k)}=a_{k}^{*}B_1^{-k} + a_k B_2^{-k}.
\end{equation}
Let us notice here that $a_{d-k}=a_{k}^{*}$ and therefore $C_i^{(d-k)}=[C_{i}^{(k)}]^{\dagger}$ for 
any $k=1,\ldots,d-1$ and $i=1,2$.

Introducing finally the following matrices 
\begin{equation}\label{ZdTd}
Z_d=\sum_{i=0}^{d-1}\omega^i\proj{i},\qquad T_d=\sum_{i=0}^{d-1}\omega^{i+\frac{1}{2}}\proj{i}-\frac{2}{d}\sum_{i,j=0}^{d-1}(-1)^{\delta_{i,0}+\delta_{j,0}}\omega^{\frac{i+j+1}{2}}|i\rangle\!\langle j|.
\end{equation}
we can now state and prove our main theorem.

\begin{thm*}
Assume that the SATWAP Bell inequality (\ref{app:SATWAP}) is maximally violated by a state $\ket{\psi}\in\mathcal{H}_{A}\otimes\mathcal{H}_B$ and observables $A_i,B_i$ with $i=1,2$ acting on $\mathcal{H}_A$ and $\mathcal{H}_B$, respectively. Then, the following statements hold true for any $d$:
\begin{enumerate}

\item There exist local unitaries $U_A: \mathcal{H}_A \rightarrow \C^d \otimes \mathcal{H}_{A'}$ and $U_B: \mathcal{H}_B \rightarrow \C^d \otimes \mathcal{H}_{B'}$ such that 
\begin{equation} \label{AiBi}
U_B B_1 U^{\dagger}_B = Z_d \otimes \I_{B'}, \qquad U_B B_2 U^{\dagger}_B = T_d \otimes \I_{B'}, 
\end{equation}
and
\begin{equation}\label{AiBi2}
U_A A_1 U^{\dagger}_A =  \left[ a_1^* Z_d + 2(a^*_1)^3 T_d \right] \otimes \I_{A'} \ , \qquad  U_A A_2 U^{\dagger}_A =  \left( a_1 Z_d - a^*_1 T_d \right) \otimes \I_{A'} \ ,
\end{equation} 
where $a_1$ is defined in Eq. (\ref{app:ak}), whereas $Z_d$ and $T_d$ are defined in Eq. \eqref{ZdTd}; 
\item  The state is $\ket{\psi}$ equivalent to the maximally entangled state of two qudits $\ket{\phi^+_d}$ in the sense that
\be \label{State}
U_A\otimes U_B \ket{\psi} = \ket{\phi^+_d} \otimes \ket{\tilde{\psi}}, 
\ee
where $\I_{A'}$ and $\I_{B'}$ are identity operators acting on  $\mathcal{H}_{A'}$ and $\mathcal{H}_{B'}$, respectively, and $\ket{\tilde{\psi}}$ is some state from $\mathcal{H}_{A'}\otimes \mathcal{H}_{B'}$.
\end{enumerate}
\end{thm*}
\begin{proof}The proof consists of three major steps to prove the theorem. First, we establish that the observables measured by Bob act on a Hilbert space of dimension multiple of $d$ and simultaneously obtain their explicit form. Second, exploiting the symmetry of the Bell-operator we recover the desired form of $A_i$ and subsequently characterize the shared entangled state.

\textit{\textbf{Observables.}} We begin our proof by showing that with the aid of the sum-of-squares decomposition (\ref{SOSR1}) one can prove that the following sets of conditions hold true:
\begin{equation}\label{Ccond1}
    C_i^{(k)}=\left[C_i^{(1)}\right]^k
\end{equation}
and
\begin{equation}\label{Ccond2}
    C_{i}^{(d-k)}C_{i}^{(k)}=\I
\end{equation}
for $i=1,2$ and $k=1,\ldots,d-1$.

To prove (\ref{Ccond1}) we first notice that for the state 
$\ket{\psi}$ and observables $A_i$ and $B_i$ the sum-of-squares decomposition 
(\ref{SOSR1}) implies that
\begin{equation}\label{ACcondition}
    A_{i}^k\otimes C_{i}^{(k)}\ket{\psi}=\ket{\psi}
\end{equation}
are satisfied for $i=1,2$ and $k=1,\ldots,d-1$. Exploiting the fact that $A_i$ are unitary
we can rewrite the above as
\begin{equation}
    \I\otimes C_{i}^{(k)}\ket{\psi}=(A_{i}^{k})^{\dagger}\otimes \I\ket{\psi}=(A_{i}^{\dagger})^{k}\otimes \I\ket{\psi}.
\end{equation}
We then make use of (\ref{ACcondition}) for $k=1$ to see that $A_{i}^{\dagger}\otimes \I\ket{\psi}=\I\otimes C_{i}^{(1)}\ket{\psi}$, which allows us to rewrite the above relation as
\begin{equation}
    \I\otimes C_{i}^{(k)}\ket{\psi}=\I \otimes \left[C_{i}^{(1)}\right]^{k}\ket{\psi},
\end{equation}
which is equivalent to the relation (\ref{Ccond1}) on the support of the second subsystem of $\ket{\psi}$.

To prove (\ref{Ccond2}) we again consider relation (\ref{ACcondition}) and 
apply $A_i^{d-k}\otimes C_i^{(d-k)}$ to it. This due to the fact that $A_i$ is unitary gives
\begin{equation}
    \I\otimes C_i^{(d-k)}C_i^{(k)}\ket{\psi}=\ket{\psi},
\end{equation}
which is equivalent to (\ref{Ccond2}) on the support of the second subsystem of $\ket{\psi}$.

In what follows we show that the conditions (\ref{Ccond1}) and (\ref{Ccond2})
are enough to fully characterize Bob's observables. First, in Lemma 
\ref{le:traceless} we show that these conditions imply that 
\begin{equation}
    \Tr(B_1^n)=\Tr(B_2^n)=0
\end{equation}
for any positive integer $n$ which is a proper divisor of $d$. Second, Lemma \ref{le:pol} tells us that the multiplicities of the eigenvalues of each of Bob's 
observables are the same. This together with the fact that $B_i$ have $d$ different eigenvalues allows us to draw two conclusions. 
First, the dimension of Bob's Hilbert space $\mathcal{H}_B$ is multiple of the number of outcomes $d$, meaning that 
\begin{equation}
\mathcal{H}_{B}=\mathbb{C}^d\otimes\mathcal{H}_{B'},    
\end{equation}
where $\mathcal{H}_{B'}$ is some auxiliary Hilbert space of unknown but
finite dimension. Second, there exists a unitary matrix $V_B:\mathcal{H}_{B}\rightarrow \mathcal{H}_{B}$ such that 
\begin{equation}
    \widetilde{B}_1:=V_BB_1V_B^{\dagger}=Z_d\otimes\I_{B'},
\end{equation}
where $Z_d$ is defined in Eq. (\ref{ZdTd}) and $\I_{B'}$ is the identity acting on $\mathcal{H}_{B'}$. Thirdly and finally, Lemma \ref{le:Fij} implies that 
there exists another unitary operation $\overline{V}_B$ that preserves $\widetilde{B}_1=Z_d\otimes \I_{B'}$
and 
\begin{equation}
    \overline{V}_B (V_B B_2V_B^{\dagger})\overline{V}_B^{\dagger}=T_d\otimes \I_{B'}
\end{equation}
with $T_d$ defined in Eq. (\ref{ZdTd}). Thus, there exist a unitary operation $U_B:\mathcal{H}_{B}\to \mathcal{H}_{B}$ given by $U_{B}=\overline{V}_BV_B$ such that (\ref{AiBi}) hold true. This completes the first part of the proof.

\textit{\textbf{Alice's observables.}} Let us now concentrate on Alice's observables. To determine them we will exploit the symmetries of the SATWAP Bell inequality. Precisely, let with the aid of the facts that $a_{d-k}=a_k^{*}$ and that $A_i^d=B_i^d=\I$,
the SATWAP Bell expression $\mathcal{I}_{d}$ can also be written as
\begin{equation}
    \mathcal{I}_d=\sum_{k=1}^{d-1}\left(\overline{C}_1^{(k)}\otimes B_{1}^{k}+\overline{C}_2^{(k)}\otimes B_2^{k}\right),
\end{equation}
where 
\begin{equation} \label{cit}
\overline{C}_1^{(k)}=
a_{k}^*A_1^{-k}+a_{k}A_2^{-k},\qquad
\overline{C}_2^{(k)}=
\omega^{-k}a_{k}A_1^{-k}+a_{k}^*A_2^{-k}.
\end{equation}
Moreover, one can construct another sum-of-squares
decomposition for the SATWAP Bell inequality in terms of these combinations
of Alice's observables. Precisely, one has
\begin{equation}\label{SOSA}
    \beta_{Q}^d\mathbbm{1}-\mathcal{B}_d=\frac{1}{2}\sum_{k=1}^{d-1}\left(\overline{P}_{1,k}^{\dagger}\overline{P}_{1,k}+\overline{P}_{2,k}^{\dagger}\overline{P}_{2,k}\right),
\end{equation}
where
\begin{equation}
    \overline{P}_{i,k}=\I-\overline{C}_i^{(k)}\otimes B_i^k.
\end{equation}
Thus, using the same reasoning as in the case of Bob's observables one can 
show that the $\overline{C}_i^{(k)}$ operators must satisfy the same relations as
the $C_i^{(k)}$ operators, that is,
\begin{equation}\label{Cbar}
    \overline{C}_i^{(k)}=\left[\overline{C}_i^{(1)}\right]^k
\end{equation}
as well as
\begin{equation}
    \overline{C}_i^{(d-k)}\overline{C}_i^{(k)}=\I
\end{equation}
with $i=1,2$ and $k=1,\ldots,d-1$. The crucial step now is to realize that 
$\overline{C}_1^{(k)}$ and $\overline{C}_2^{(k)}$ have the same 
form as, respectively, $C_2^{(k)}$ and $C_1^{(k)}$ and hence 
from the whole analysis performed for Bob's observables 
we know that 
\begin{equation}
    \mathcal{H}_A=\mathbb{C}^d\otimes\mathcal{H}_{A'}
\end{equation}
with $\mathcal{H}_{A'}$ being some auxiliary Hilbert space corresponding to 
the party $A$ of unknown but finite dimension. Moreover, there exists a unitary operation
$V_A: \mathcal{H}_{A}\to \mathcal{H}_{A}$ such that 
\begin{equation}
    \widetilde{A}_1:=V_A A_1V^{\dagger}_A=Z_d\otimes \I_{A'},\qquad \widetilde{A}_2:=V_A A_2 V^{\dagger}_A=T_d\otimes \I_{A'}.
\end{equation}
Thus, similarly to Bob's observables we can fully characterize Alice's measurements. In principle we could end our characterization of observables already here, however, for further benefits we will apply another local unitary operation $W_A:\mathbb{C}^d\to\mathbb{C}^d$ 
to the qudit part of $\widetilde{A}_i$ observables so that 
\begin{equation}\label{dupa1}
    W_A Z_d W_A^{\dagger}=a_1^{*}Z_d+2(a_1^{*})^3 T_d
\end{equation}
and
\begin{equation}\label{dupa2}
       W_A T_d W_A^{\dagger}=a_1Z_d-a_1^{*}T_d.
\end{equation}
The existence of $W_A$ is shown under Fact \ref{fact:cglmp}-\ref{fact:CtoZ} in Appendix \ref{sec:cglmp}. Thus, we have proven the existence of a unitary operation $U_A=(W_A\otimes\I_{A'}) V_A$ such that Eq. (\ref{AiBi2}) hold true. Importantly, by combining Eqs. (\ref{dupa1}), (\ref{dupa2}) and (\ref{Cbar}), application of the unitary $U_A$ to $\overline{C}_{i}^{(k)}$ operators allows us to bring them to 
\begin{equation}
    U_A \overline{C}_1^{(k)}U_A^{\dagger}=(Z_d^{\dagger})^k\ot \I_{A'},\qquad U_A \overline{C}_2^{(k)}U_A^{\dagger}=(T_d^{\dagger})^k\ot \I_{A'}.
\end{equation}
In particular, if Alice's observables are transformed by applying the unitary $U_A$, then $\overline{C}_1^{(1)}=Z_d^*$ and $\overline{C}_2^{(1)}=T_d^*$.
As we will see next this choice of the $A_i$ observables as well as the $\overline{C}_i^{(1)}$ operators guarantees that the self-tested state (up to some additional degrees of freedom) is the maximally
entangled state of two qudits written in the standard basis
\begin{equation}
    \ket{\phi_{d}^+}=\frac{1}{\sqrt{d}}\sum_{i=0}^{d-1}\ket{ii}.
\end{equation}

\textit{\textbf{State.}} Let us finally prove Eq. (\ref{State}). From the fact that $\ket{\psi}$ violates maximally the SATWAP Bell inequality we know, by virtue of the SOS decomposition (\ref{SOSA}) that 
\begin{equation}
    \overline{C}^{(k)}_1\otimes B_1^k\ket{\psi}=\ket{\psi},
\end{equation}
with $i=1,2$ and $k=1,\ldots,d-1$, 
which after application of $U_A\otimes U_B$ can be rewritten as
\begin{equation}\label{cc1}
    [(Z_d^{*})^{(k)}\otimes Z_d^{k} \otimes \I_{A'B'}]\ket{\psi}=\ket{\psi}
\end{equation}
and
\begin{equation}\label{cc2}
    [(T_d^{*})^{(k)}\otimes T_d^{k} \otimes \I_{A'B'}]\ket{\psi}=\ket{\psi},
\end{equation}
where by $\I_{A'B'}$ we denote the identity acting on $\mathcal{H}_{A'}\otimes \mathcal{H}_{B'}$. 

Since the state $\ket{\psi}$ belongs to $\mathbb{C}^d\otimes\mathbb{C}^d\otimes\mathcal{H}_{A'}\otimes\mathcal{H}_{B'}$ we can decompose it as 
\begin{equation}
\ket{\psi}=\sum^{d-1}_{i,j=0}\ket{ij}\ket{\psi_{ij}},
\end{equation}
where $\ket{\psi_{ij}}\in\mathcal{H}_{A'}\otimes \mathcal{H}_{B'}$ are some in general nonorthogonal and not normalized states. With this decomposition the condition \eqref{cc1} with $k=1$ implies that 
\begin{equation}
\sum_{i,j=0}^{d-1}\omega^{j-i}\ket{ij}\ket{\psi_{ij}}=\sum_{i,j=0}^{d-1}\ket{ij}\ket{\psi_{ij}}.
\end{equation}
This condition directly implies that 
\be \label{sf1}
\ket{\psi_{ij}}=0
\ee 
for any pair $i,j=0,\ldots,0$ such that $i\neq j$.

Let us now exploit the other condition \eqref{cc2} with $k=1$, which after projecting onto 
the ket $\langle ii|$ gives
\begin{eqnarray}
\ket{\psi_{ii}}=\sum_{k=0}^{d-1}\langle i|T_d^*|k\rangle\langle i|T_d|k\rangle\ket{\psi_{kk}}=\sum_{k=0}^{d-1}|\langle i|T_d|k\rangle|^2\ket{\psi_{kk}}
\end{eqnarray}
where we have used the fact $\langle i|T_d^*|k\rangle=\langle k|T_d^*|i\rangle$ as $T_d$ is symmetric. Using the explicit form of $T_d$ given in \eqref{ZdTd} we can 
rewrite the above as
\begin{eqnarray}
\ket{\psi_{ii}}=\frac{4}{d^2}\sum^{d-1}_{k=0}\ket{\psi_{kk}}+\left[\left(1-\frac{2}{d}\right)^2-\frac{4}{d^2}\right]\ket{\psi_{ii}} = \frac{4}{d^2}\sum^{d-1}_{k=0}\ket{\psi_{kk}}+\left[1-\frac{4}{d}\right]\ket{\psi_{ii}}. 
\end{eqnarray}
We now move the $\ket{\psi_{ii}}$ ket to the left-hand side of
the equality (for $d=4$ there is nothing to move) and obtain
the following relation,
\begin{equation} \label{sf2}
\forall i, \quad \frac{4}{d}\ket{\psi_{ii}}=\frac{4}{d^2}\sum^{d-1}_{k=0}\ket{\psi_{kk}},
\end{equation}
from which it follows that all $\ket{\psi_{ii}}$ are equal. 
This together with (\ref{sf1}) allows us to conclude that 
\begin{equation}
    U_A\otimes U_B\ket{\psi}=\frac{1}{\sqrt{d}}\sum_{i=0}^{d-1}\ket{ii}\otimes \ket{\widetilde{\psi}}
\end{equation}
with $\ket{\widetilde{\psi}}=\sqrt{d}\ket{\psi_{ii}}$. We thus recover (\ref{State}), completing the proof.
\end{proof}


Below we append three lemmas used in the proof of the above theorem. Remark that, series of Observations are employed inside the proof of Lemma \ref{le:traceless} and Lemma \ref{le:Fij}.

\begin{lem}\label{le:traceless}
Consider two unitary observables $B_i \ (i\in \{1,2\})$ acting on a finite-dimensional Hilbert space whose eigenvalues are $\w^l \ (l \in \{0,\dots,d-1\})$. If they satisfy the conditions \eqref{Ccond1} and \eqref{Ccond2}, then for any proper divisor $n$ of $d$, that is, a positive integer such that $n< d$ and $d/n\in\mathbb{N}$,
\be \label{traceB}
\Tr(B^n_1) = \Tr(B^n_2) = 0.
\ee 
\end{lem}
\begin{proof}Plugging the explicit form of $C_2^{(k)}$ given in Eq. (\ref{Coperators}) into Eq. (\ref{Ccond2}) and using the fact that $B_i^{d-k}=B_i^{-k}$, we arrive at
\be 
(a_{d-k}^{*}B_1^{k} + a_{d-k} B_2^{k})  (a_{k}^{*}B_1^{-k} + a_k B_2^{-k}) = \I.
\ee
Using then the fact that $a^*_k = a_{d-k}$ and $a^*_k/a_k = \mathbbm{i}\omega^{\frac{k}{2}}$, a simple calculation leads us to the following condition
\begin{equation}\label{Obs22}
B_1^{k}B_2^{-k}=\omega^{-k}B_2^{k}B_{1}^{-k}
\end{equation} 
with $k=1,\ldots,d-1$. Due to the fact that $B^d_i= \I$, the above relation \eqref{Obs22} extends to any integer $k \in \mathbb{Z}$. We now establish some relations under observations \ref{fact:4s} and \ref{fact:ObsTraces} which will be used later in the proof. 
\begin{obs} \label{fact:4s}
First, we show that the following identities hold true for any non-negative integers $s,x,y\in \mathbb{N} \cup \{0\}$
\begin{equation}\label{NewId1}
    \Tr(B_1^x)=\omega^{sx}\,\Tr\left(B_1^{(2s+1)x}B_2^{-2sx}\right).
\end{equation}
and
\begin{equation}\label{NewId2}
    \Tr(B_2^y)=\omega^{sy}\,\Tr\left(B_1^{2sy}B_2^{(-2s+1)y}\right).
\end{equation}
%
\end{obs}
We present only the proof of the first identity as that of the second one is analogous. 
It is pretty straightforward and consists of multiple application of the identity (\ref{Obs22}).
Let us thus focus on the right-hand side of Eq. (\ref{NewId1}) and 
consider (\ref{Obs22}) for $k=2sx$, multiply it by $B_1^x$ and trace both 
sides of the resulting equation. This gives
\begin{equation}\label{Eq1}
    \Tr\left(B_1^{(2s+1)x}B_2^{-2sx}\right)=\omega^{-2sx}\Tr\left(B_2^{2sx}B_1^{(-2s+1)x}\right).
\end{equation}
We consider again (\ref{Obs22}) for $k=(2s-1)x$, 
multiply it by $B_2^x$ and then trace both sides, which results in
\begin{equation}\label{Eq2}
    \Tr\left(B_2^{2sx}B_1^{-(2s-1)x}\right)=\omega^{(2s-1)x}\left(B_1^{(2s-1)x}B_{2}^{-(2s-2)x}\right).
\end{equation}
After plugging Eq. (\ref{Eq2}) into Eq. (\ref{Eq1}) we obtain
\begin{equation}\label{Eq3}
    \Tr\left(B_1^{(2s+1)x}B_2^{-2sx}\right)=\omega^{-x}\left(B_1^{(2s-1)x}B_{2}^{-(2s-2)x}\right).
\end{equation}
We have thus lowered the power of $B_1$ from $(2s+1)x$ in Eq. (\ref{Eq1}) 
to $(2s-1)x$ in Eq. (\ref{Eq3}). We repeat this double substitution until the 
power of $B_1$ is $x$, each time acquiring the phase $\omega^{-x}$ 
(notice that in this case the power of $B_2$ goes from $-2sx$ in 
Eq. (\ref{Eq1}) to zero). As a result we arrive at
\begin{equation}
    \Tr\left(B_1^{(2s+1)x}B_2^{-2sx}\right)=\omega^{-sx}\Tr(B_1^x),
\end{equation}
which is what we wanted to show.

\begin{obs}\label{fact:ObsTraces}
Next, the following simple relation between traces of powers of $B_i$ operators holds true
\begin{equation}\label{traces}
\Tr(B_1^x)=\omega^{-\frac{x}{2}}\Tr(B_2^x), \qquad x = 1,\dots, \left\lfloor{\frac{d}{2}}\right\rfloor.
\end{equation}
\end{obs}
In order to prove this observation let us consider \eqref{Ccond1} for $k=2x$, which we can equivalently state as
\begin{equation}
    \left[C_2^{(2x)}\right]^{\dagger}=\left\{\left[C_2^{(x)}\right]^2\right\}^{\dagger}.
\end{equation}
This after taking into account \eqref{Coperators} implies
\begin{equation}
\omega^{\frac{x}{2}}B_1^{2x}+\omega^{-\frac{x}{2}}B_2^{2x}=\{B_1^x,B_2^x\}. 
\end{equation}
We then multiply the above by $B_1^{-x}$ and by taking trace on both sides we find
\begin{equation}\label{Tinto}
\omega^{\frac{x}{2}}\Tr(B_1^x)+\omega^{-\frac{x}{2}}\Tr(B_1^{-x}B_2^{2x})=2\Tr(B_2^x).
\end{equation}
We then use \eqref{Obs22} with $k=x$, i.e., $B_2^xB_1^{-x}=\omega^x B_1^xB_2^{-x}$.
By multiplying it by $B_2^x$ and tracing both sides we obtain 
\begin{equation} \label{bb}
\Tr(B_1^{-x}B_2^{2x})=\omega \Tr(B_1^x).
\end{equation}
Further, substituting the expression of $\Tr(B_1^{-x}B_2^{2x})$ from Eq. \eqref{bb} into Eq. \eqref{Tinto} leads us to 
\eqref{traces}. 

Let us now make use of the above observations for finally prove the validity of \eqref{traceB}. 
We consider the cases of even and odd $d$ separately. 

\textit{\textbf{Even} $d$.} 
Let $n$ be a divisor of $d$, i.e., $d/n \in \mathbb{N}$. There are two possibilities, 
$d/n$ is even or odd, i.e., there exists integer $s$ such that $n = d/(2s)$ or 
$n = d/(2s+1)$, respectively. Whenever $d/n$ is even, we 
consider Eq. \eqref{NewId1} for $x = n = d/(2s)$, which reads
\begin{eqnarray}\label{traces22}
\Tr(B_1^{n})=\omega^{d/2}\,\Tr\left(B_1^{d+n}B_2^{-d}\right).
\end{eqnarray}
Using then the facts that $B_i^d=\I$ and $\omega^{d/2}=-1$, the above relation simplifies to,
\be 
\Tr(B_1^{n})= (-1) \Tr\left(B_1^{n}\right).
\ee 
Thus, for any $n$ such that $d/n$ is even, $\Tr(B_1^{n})=0$. 
Similarly, by taking $y=d/(2s)$ in Eq. (\ref{NewId2}) one can conclude the same for $B_2$, i.e., $\Tr(B_2^{n})=0$. 

Now, for any divisor $n$ of $d$ for which $d/n$ is odd, we choose $x = n = d/(2s+1)$ in Eq. \eqref{NewId1}, 
which leads us to
\be \label{even-1}
\Tr(B_1^{n})= (-1) \w^{-n/2}\, \Tr\left(B_2^{n}\right).
\ee
The above relation together with Eq. \eqref{traces} imply $\Tr(B_i^{n})=0$ for any $n$ such that $d/n$ is odd 
and $n\leq d/2$. Thus, we have shown that for any $n$ which is a divisor of $d$ and $n\neq d$, $\Tr(B_i^{n})=0$ for $i=1,2$.

\textit{\textbf{Odd} $d$.}
Since $d$ is odd, there exists an integer $s$ such that $2s+1=d$. Plugging $s=(d-1)/2$ in Eq. \eqref{NewId1} we get
\begin{eqnarray}\label{traces21}
\Tr(B_1^x)=\omega^{(d-1)x/2}\,\Tr\left(B_1^{dx}B_2^{-(d-1)x}\right),
\end{eqnarray}
which due to the facts that $B_2^d=\I$ and $\omega^{d/2}=-1$ simplifies to
\begin{eqnarray}\label{traces2}
\Tr(B_1^x)=(-1)^x\omega^{-x/2}\,\Tr\left(B_2^{x}\right).
\end{eqnarray}
Now, for any odd $x$ such that $x\leq \lfloor d/2\rfloor$ we obtain from Eqs. \eqref{traces} and \eqref{traces2} that $\Tr(B_1^x)=\Tr(B_2^x)=0$. 
Since all the divisors of an odd $d$ (except $d$) are odd and less than $\lfloor d/2\rfloor$, we conclude \eqref{traceB}.
\end{proof}

\begin{lem} \label{le:pol}
Consider a real polynomial 
\begin{equation}\label{WX}
    W(x)=\sum_{i=0}^{d-1}\lambda_ix^i
\end{equation}
with rational coefficients $\lambda_i \in \Q$. Assume that $\omega^n$ with $\w=\mathrm{e}^{2\pi \mathbbm{i}/d}$ 
is a root of $W(x)$ for any $n$ being a proper divisor of $d$, i.e., $n\neq d$ such that 
$d/n\in \N$. Then, $\lambda_0=\lambda_1=\ldots=\lambda_{d-1}$.
\end{lem}

To be able to prove this lemma we need to introduce the following two relevant concepts. \\

\noindent {\bf Euclidean division algorithm.} Let $F(x), G(x) \in \mathbb{F}[x]$ be two nonzero polynomials over the field 
$\mathbb{F}$ such that $\deg G(x)\leq \deg F(x)$. Then there exist two polynomials $Q(x), R(x) \in \mathbb{F}[x]$ such that $F(x) = Q(x) G(x)+R(x)$ and $\deg R(x) < \deg G(x)$.

\

\noindent\textbf{Cyclotomic polynomials.} Let $n$ be any positive integer. The $n$th cyclotomic polynomial $\phi_n(x)$ is defined as 
the unique irreducible polynomial over the field of rational numbers whose root is the primitive root of unity $\mathrm{exp}(2\pi\mathbbm{i}/n)$, and is given explicitly by the following formula
\begin{eqnarray}
\phi_{n}(x)=\prod_{\substack{1\leq k\leq n\\\mathrm{gcd}(k,n)=1}}\left(x-e^{\frac{2\pi \mathbbm{i}k}{n}}\right).
\end{eqnarray}
Let us then consider a positive integer $d$ and denote by $n_1<n_2<\ldots<n_k$
its proper divisors with $n_1=1$ and $n_k<d$. Then, it is known that the product of the cyclotomic polynomials $\phi_{d/n_i}$ over all these proper divisors can be 
represented by the following simple formula
\begin{equation}\label{FundamentalCyclotomic0}
    \prod_{i=1}^{k}\phi_{d/n_i}(x)=\frac{x^d-1}{x-1}\equiv \sum_{i=0}^{d-1}x^i.
\end{equation}
Let us finally notice for further purposes that $\omega^{n_i}$ for some
divisor $n_i$ is a root of $\phi_{d/n_i}(x)$ but not of $\phi_{d/n_j}(x)$
for any other divisor $n_j\neq n_i$ of $d$.

\begin{proof}Let $n_1< n_2<n_3<\ldots<n_k<d$ with $n_1=1$ and $n_k<d$ be the proper divisors of $d$. From the assumption we know that $\omega^{n_i}$ is a root of $W(x)$, 
that is 
\begin{equation}\label{Roots}
    W(\omega^{n_i})=0,\qquad i=1,\ldots,k.
\end{equation}
We now use this last condition along with the Euclidean division algorithm 
to prove that the polynomial $W(x)$ must be a product of cyclotomic polynomials $\phi_{d/n_i}(x)$. We do it in a recursive way 

Let us begin with $n_1=1$. The Euclidean division algorithm implies that 
\begin{equation}
W(x)=Q_1(x)\phi_{d}(x)+R_1(x)    
\end{equation}
with $Q_1(x),R_1(x)\in\mathbb{Q}[x]$ such that $\deg R_1(x)<\deg \phi_d(x)$.
From Eq. (\ref{Roots}) we know that $W(\omega)=0$, implying that $R_1(\omega)=0$. 
As $\phi_d(x)$ is the unique polynomial of minimal degree over the rational field whose root is $\omega$, the latter implies that $R_1(x)\equiv 0$. Thus, 
\begin{equation}
    W(x)=Q_1(x)\phi_d(x).
\end{equation}

Let us then move on to the next divisor of $d$, $n_2>n_1$. Due to the fact that $\omega^{n_2}$ is not a root of $\phi_{d}(x)$, the condition (\ref{Roots}) implies that $Q_{1}(\omega^{n_2})=0$. We therefore apply the Euclidean division algorithm to $Q_1(x)$, that is,
\begin{equation}
    Q_1(x)=Q_{2}(x)\phi_{d/n_2}(x)+R_{2}(x)
\end{equation}
with $Q_{2}(x),R_{2}(x)\in\mathbb{Q}[x]$ such that $\deg R_{2}<\deg \phi_{d/n_2}(x)$.
Due to the fact that $\omega^{n_2}$ is a root of $Q_1$ we see that $R_{2}(\omega^{n_2})=0$. 
Again, due to the fact that $\phi_{d/n_2}(x)$ is the unique polynomial of minimal degree over the rational field whose root is $\omega^{n_2}$, the latter implies $R_{2}(x)\equiv 0$. Consequently, 
\begin{equation}
    Q_1(x)=Q_2(x)\phi_{d/n_2}(x).
\end{equation}

Applying the above procedure iteratively for every divisor $n_i$ of $d$ it is now
not difficult to see that our initial polynomial $W(x)$ is of the following form 
\begin{equation}\label{WX2}
    W(x)=\widetilde{Q}(x)\prod_{i=1}^{k}\phi_{d/n_i}(x),
\end{equation}
with $\widetilde{Q}(x)$ is a rational polynomial. Owing to the fundamental relation satisfied by the cyclotomic polynomials given in Eq. (\ref{FundamentalCyclotomic0})
%
%
one sees that the degree of the above product of cyclotomic polynomials is $d-1$ 
and equals the degree of $W(x)$, which directly implies that $\widetilde{Q}(x)$ is 
simply a constant factor, denoted $\widetilde{Q}$. Moreover, by comparing Eqs. (\ref{WX})
and (\ref{WX2}) and taking into account Eq. (\ref{FundamentalCyclotomic0}), we immediately 
see that $\lambda_i=\widetilde{Q}$ for $i=0,\ldots,d-1$, meaning that 
$\lambda_0=\lambda_1=\ldots=\lambda_{d-1}$, which completes the proof.
\end{proof}


\begin{lem}\label{le:Fij}
Consider two unitary operators $B_1$ and $B_2$ acting on $\C^d\otimes \mathcal{H}_{B'}$ such that $B_i^d=\I$ with $i=1,2$. Assume moreover that $B_1=Z_d \otimes \I_{B'} $
and that they both satisfy the conditions (\ref{Ccond1}).
%
Then, then there exists a unitary $U$ such that $UB_1 U^\dagger = B_1$ and 
$UB_2U^{\dagger}=T_d\otimes\mathbbm{1}_{B'}$ with $T_d$ defined in (\ref{ZdTd}).
%
\end{lem}
\begin{proof}
We begin by proving the following relation for $B_1$ and $B_2$ matrices:
\begin{equation}\label{FijEq1}
B_2^k=-(k-1)\omega^{\frac{k}{2}}B_1^k+\omega^{\frac{k-1}{2}}\sum_{m=0}^{k-1}B_1^mB_2B_1^{k-1-m}, \qquad k=1,\ldots,d.
\end{equation}
We prove this relation by induction. For $k=1$ it is not difficult to see that both its sides equal $B_2$. Assuming then that (\ref{FijEq1}) holds true we will prove it for $k+1$. To this end, let us rewrite \eqref{Ccond1} as 
\begin{equation}
 \left[C_2^{(k+1)}\right]^\dagger=\left[C_2^{(1)}C_2^{(k)}\right]^\dagger   
\end{equation}
for $k=1,\dots,d-1$, which in terms of $B_i$ [cf. Eq. \eqref{Coperators}] can be rewritten as
%
\begin{eqnarray}
   B_2^{k+1}=-\omega^{\frac{k+1}{2}}B_1^{k+1}+\omega^{\frac{k}{2}}B_1^kB_2+\omega^{\frac{1}{2}}B_2^kB_1 .
\end{eqnarray}
Now, substituting $B_2^k$ from \eqref{FijEq1}, we have
\begin{eqnarray}
   B_2^{k+1} &= &-\omega^{\frac{k+1}{2}}B_1^{k+1}+\omega^{\frac{k}{2}}B_1^kB_2+\omega^{\frac{1}{2}}\left[-(k-1)\omega^{\frac{k}{2}}B_1^k+\omega^{\frac{k-1}{2}}\sum_{m=0}^{k-1}B_1^mB_2B_1^{k-1-m}\right]B_1\nonumber\\
 &= &-k\omega^{\frac{k+1}{2}}B_1^{k+1}+\omega^{\frac{k}{2}}\sum_{m=0}^{k}B_1^mB_2B_1^{k-m}.
\end{eqnarray}

Having established (\ref{FijEq1}), let us now write $B_2$ as
\be  \label{B2form}
B_2=\sum_{i,j=0}^{d-1}\ket{i}\!\bra{j}\otimes F_{ij},
\ee  
where $F_{ij}$ are some matrices acting on $\mathcal{H}_{B'}$. Our aim now is to determine $F_{ij}$ using the conditions (\ref{FijEq1}). We first focus on $F_{ii}$ and then move on to $F_{ij}$ with $i\neq j$.

\textit{\textbf{Finding} $F_{ii}$.} Taking the relation \eqref{FijEq1} for $k=d-1$, we obtain %
\begin{eqnarray}\label{FijEq3}
    B_2^\dagger=-(d-2)\omega^{\frac{d-1}{2}}B_1^{d-1}+\omega^{\frac{d-2}{2}}\sum_{m=0}^{d-2}B_1^mB_2B_1^{d-2-m},
\end{eqnarray}
which after taking into account that $B_1=Z_d\otimes \I_{B'}$ and that $B_2$ is given by Eq. (\ref{B2form}), can be rewritten as
\begin{equation}\label{FijEq9}
 \sum\limits_{i,j=0}^{d-1}\ket{j}\!\bra{i}\otimes F_{ij}^\dagger =
-(d-2)\omega^{\frac{d-1}{2}}\sum_{i=0}^{d-1}\omega^{(d-1)i}\ket{i}\!\bra{i}\otimes \I_{B'} + \omega^{\frac{d-2}{2}} \sum_{i,j=0}^{d-1}\sum_{m=0}^{d-2}\omega^{(d-2)j+m(i-j)}\ket{i}\!\bra{j}\otimes F_{ij}.
\end{equation}
By projecting the first subsystem onto $|i\rangle\!\langle i|$, we arrive at the following 
condition for $F_{ii}$ matrices
\begin{eqnarray}\label{FijEq8}
  F_{ii}^\dagger =-(d-2)\omega^{\frac{d-1}{2}} \omega^{-i}\I_{B'} +(d-1)\omega^{\frac{d-2}{2}} \omega^{-2i}F_{ii}.
\end{eqnarray}
Conjugating the above equation on both sides we obtain
\begin{eqnarray}\label{FijEq85}
     F_{ii}=-(d-2)\omega^{-\frac{d-1}{2}}\omega^{i}\ \I_{B'} +(d-1)\omega^{-\frac{d-2}{2}}\omega^{2i} F_{ii}^\dagger .
\end{eqnarray}
After plugging into the above relation $F_{ii}^\dagger$ from Eq. \eqref{FijEq8} 
we obtain an equation for $F_{ii}$ whose solution is
%
%
%
\begin{equation}\label{TintoDeVerano}
    F_{ii}=\frac{d-2}{d}\omega^{i+\frac{1}{2}} \I_{B'}.
\end{equation}

\textit{\textbf{Finding} $F_{ij}$.} Let us now move on to determining the $F_{ij}$ matrices for $i\neq j$. We formulate our derivation as a sequence of observations. 

First, by comparing the matrices appearing on $|i\rangle\!\langle j|$ position of both sides of \eqref{FijEq9} with $i\neq j$, we obtain the following equation
\begin{eqnarray}\label{fijEq12}
F_{ji}^\dagger=\omega^{\frac{d-2}{2}}\omega^{-2j}\sum_{m=0}^{d-2}\omega^{m(i-j)} F_{ij},
\end{eqnarray}
which after taking into account the fact that for $i\neq j$,
\begin{equation}
    \sum_{m=0}^{d-2}\omega^{m(i-j)}=-\omega^{-(i-j)},
\end{equation}
reduces to 
\be \label{fijEq12}
F_{ij}=\omega^{i+j+1}F_{ji}^\dagger.
\ee 
Thus, \eqref{FijEq3} only provides a relation between the symmetric elements of $B_2$ in the form \eqref{fijEq12}. To find the explicit form of $F_{ij}$, we have to look for equations involving higher order terms in $F_{ij}$. To this aim, let us prove the following observation.

\begin{obs}\label{obs3.1}
The following conditions hold true
 \begin{eqnarray} \label{FijObs3}
  -(k-1)\omega^{\frac{k}{2}}\sum_{i,j=0}^{d-1}\omega^{ki}\ket{i}\!\bra{j}\otimes F_{ij}+\omega^{\frac{k-1}{2}}\sum_{i,j=0}^{d-1}\ket{i}\!\bra{j}\otimes\left[\sum_{\substack{l=0\\l\ne i}}^{d-1}\left(\frac{\omega^{ki}-\omega^{kl}}{\omega^{i}-\omega^{l}}\right) F_{il}F_{lj}+k\omega^{(k-1)i}F_{ii}F_{ij}\right] \nonumber\\
  = -k\omega^{\frac{k+1}{2}}\sum_{i=0}^{d-1}\omega^{(k+1)i}\ket{i}\!\bra{i}\otimes \I_{B'}+\omega^{\frac{k}{2}}\sum_{i,j=0}^{d-1}\left(\sum_{m=0}^{k}\omega^{k j+m(i-j)}\ket{i}\!\bra{j}\otimes F_{ij}\right), \qquad k=1,\ldots,d-1.
\end{eqnarray}
\end{obs}
From \eqref{FijEq1} we know that 
\begin{eqnarray}\label{FijEq11}
     B_2^{k+1}=-k\omega^{\frac{k+1}{2}}B_1^{k+1}+\omega^{\frac{k}{2}}\sum_{m=0}^{k}B_1^mB_2B_1^{k-m}.
\end{eqnarray}
Again using \eqref{FijEq1} for $B_2^k$, the left-hand side of the above equation can be written as
\begin{equation} \label{FijEq11a}
B_2^{k+1}=B_2^kB_2=-(k-1)\omega^{\frac{k}{2}}B_1^kB_2+\omega^{\frac{k-1}{2}}\sum_{m=0}^{k-1}B_1^mB_2B_1^{k-1-m}B_2.
\end{equation}
Let us now focus on the sum appearing in Eq. (\ref{FijEq11a}). After substituting the explicit forms of $B_1$ and $B_2$ into it, we can rewrite it as
\begin{equation}
\sum_{m=0}^{k-1}B_1^mB_2B_1^{k-1-m}B_2 
= \sum_{i,j,l=0}^{d-1}\omega^{(k-1)l}\sum_{m=0}^{k-1}\omega^{m(i-l)}\ket{i}\!\bra{j}\otimes F_{il}F_{lj}.
\end{equation}
After splitting the above summations into the cases of $i=l$ and $i\neq l$, we obtain
\begin{eqnarray} \label{eq83}
\sum_{m=0}^{k-1}B_1^mB_2B_1^{k-1-m}B_2 & =&\sum_{i,j=0}^{d-1} \ket{i}\!\bra{j} \otimes \left[k\omega^{(k-1)i} F_{ii}F_{ij}+\sum_{\substack{l=0\\l\ne i}}^{d-1}\left(\omega^{(k-1)l}\sum_{m=0}^{k-1}\omega^{m(i-l)}F_{il}F_{lj}\right)\right]
\nonumber \\
& =& \sum_{i,j=0}^{d-1}\ket{i}\!\bra{j}\otimes\left[\sum_{\substack{l=0\\l\ne i}}^{d-1}\left(\frac{\omega^{ki}-\omega^{kl}}{\omega^{i}-\omega^{l}}\right) F_{il}F_{lj}+k\omega^{(k-1)i}F_{ii}F_{ij}\right].
\end{eqnarray}

Let us then consider the sum appearing in Eq. (\ref{FijEq11}). 
After plugging the explicit forms of $B_1$ and $B_2$ into it we arrive at 
\begin{equation}  \label{eq84}
\sum_{m=0}^{k}B_1^mB_2B_1^{k-m} = \sum_{i,j=0}^{d-1}\omega^{kj}\sum_{m=0}^{k}\omega^{m(i-j)}\ket{i}\!\bra{j}\otimes F_{ij}.
\end{equation} 
We finally substitute $B_2^{k+1}$ from Eq. \eqref{FijEq11} into Eq. \eqref{FijEq11a}, and then plug the explicit forms of $B_1$ and $B_2$ as well as the expressions \eqref{eq83} and
\eqref{eq84} into the resulting equations, which leads us to \eqref{FijObs3}. This completes the proof of observation \ref{obs3.1}.

\ \\

Now, having the formula (\ref{FijObs3}), we proceed as before. That is, we first look at the diagonal elements of it and solve the resulting equations. Then, we consider the off-diagonal elements of Eq. (\ref{FijObs3}) to have full characterization of $F_{ij}$. 
As for the diagonal part of Eq. (\ref{FijObs3}) we establish the following observation.

\begin{obs} The following conditions hold true
\begin{equation} \label{FijFact2}
\sum_{\substack{j=0\\j\ne i}}^{d-1}
\left(\frac{1-\omega^{k(j-i)}}{1-\omega^{i-j}}\right)F_{ij}F_{ij}^\dagger=\frac{4k}{d^2}\I_{B'}, \qquad k=0,\ldots,d-1, \qquad i=0,\ldots,d-1.
\end{equation}
\end{obs}
\noindent To prove the above relation let us consider the $|i\rangle\!\langle i|$ elements of \eqref{FijObs3} which, after some algebra, gives us the following relation
\begin{eqnarray}
 \sum_{\substack{l=0\\l\ne i}}^{d-1}\left(\frac{\omega^{ki}-\omega^{kl}}{\omega^{i}-\omega^{l}}\right) F_{il}F_{li} =
 k\omega^{ki}\left[2\omega^{\frac{1}{2}}F_{ii}-\omega^{-i}F_{ii}^2-\omega^{i+1}\I_{B'}\right],
 %
\end{eqnarray}
which after substituting $F_{ii}$ from Eq. \eqref{TintoDeVerano} into it simplifies to 
\begin{equation}
 \sum_{\substack{l=0\\l\ne i}}^{d-1}\left(\frac{\omega^{ki}-\omega^{kl}}{\omega^{i}-\omega^{l}}\right) F_{il}F_{li} =-\frac{4k}{d^2}\omega^{i(k+1)+1}\I_{B'}.
 \end{equation}
%
%
This after some manipulations can further be rewritten as
\begin{equation}
\sum_{\substack{l=0\\l\ne i}}^{d-1}\left(\frac{1-\omega^{k(l-i)}}{1-\omega^{i-l}}\right)F_{il}F_{li}\omega^{-(i+l+1)}=\frac{4k}{d^2} \I_{B'}, 
\end{equation}
and after taking into account \eqref{fijEq12} and changing the index $l$ to $j$, we obtain the desired relation \eqref{FijFact2}.\\

%
The subsequent observation provides the solution of equations (\ref{FijFact2}).
\begin{obs}
The solution of \eqref{FijFact2} is given by
\begin{eqnarray}\label{FijEq25}
  F_{ij}F_{ij}^\dagger=\frac{4}{d^2} \I_{B'}, \qquad i \neq j.
\end{eqnarray}
\end{obs}
We multiply Eq. (\ref{FijFact2}) by $\omega^{kn}$ with $k=0,\ldots,d-1$ and
$n=1,\ldots,d-1$ and then sum the resulting identity over all $k'$s, which leads us to 
\begin{equation}\label{formulaTwo}
-\sum_{\substack{j=0\\j\ne i}}^{d-1}\frac{1}{1-\omega^{i-j}}F_{ij}F_{ij}^\dagger\sum_{k=0}^{d-1}\omega^{k(j-i+n)}=\frac{4}{d^2}\I_{B'} \sum_{k=0}^{d-1}k\omega^{kn},
\end{equation}
where we have exploited the fact that for any $n=1,\ldots,d-1$,
\begin{equation}\label{Porto}
    \sum_{k=0}^{d-1}\omega^{kn}=0.
\end{equation}
Applying then Eq. (\ref{FijIden1}) to the right-hand side of Eq. (\ref{formulaTwo}) 
and the fact that 
\begin{equation}
    \sum_{k=0}^{d-1}\omega^{k(j-i+n)}=\delta_{j,i-n \mod d}
\end{equation}
to its left-hand side we can bring it to 
\begin{eqnarray}
F_{i(i-n \mod d)}F_{i(i-n\mod d)}^\dagger=\frac{4}{d^2}\I_{B'}.
\end{eqnarray}
To complete the proof it suffices to realize that for any 
$i=0,\ldots,d-1$ there exist $n=1,\ldots,d-1$ such that $i-n\mod d$ is any number from $\{0,\ldots,d-1\}$ different than $i$.

\ \\

The relation \eqref{FijEq25} is unfortunately not sufficient to fully characterize $F_{ij}$. However, we can find a unitary operation $U$ that preserves $B_1$ and allows us to 
obtain an explicit form of a few matrices $F_{ij}$ from Eq. \eqref{FijEq25}. 

To be more precise, let us consider a unitary matrix $U$ acting on $\C^d \otimes \mathcal{H}_{B'}$ of the form
\begin{eqnarray}
U = \sum^{d-1}_{i,j=0} \ket{i}\!\bra{j} \otimes U_i, 
\end{eqnarray} 
where
\begin{equation}
    U_0=\I_{B'},\qquad U_i=\frac{d}{2}\omega^{-\frac{i+1}{2}}F_{0i}, \qquad i=1,\ldots,d-1.
\end{equation}
We can readily check that $UB_1 U^\dagger = B_1$. Let us then denote
\begin{eqnarray}
 UB_2U^\dagger = \sum^{d-1}_{i,j=0}\ket{i}\!\bra{j} \otimes U_iF_{ij}U_j^\dagger : = \sum^{d-1}_{i,j=0} \ket{i}\!\bra{j} \otimes \tilde{F}_{ij}.
\end{eqnarray}
Note that, all the algebraic relations for $F_{ij}$ obtained so far hold also for $\tilde{F}_{ij}$, and $\tilde{F}_{ii} = F_{ii}$. 

Now, we see that
\begin{eqnarray}\label{FijEq28}
  \tilde{F}_{0j}=U_0F_{0j}U_j^\dagger=\frac{d}{2}\omega^{\frac{j+1}{2}}F_{0j}F_{0j}^\dagger=\frac{2}{d}\omega^{\frac{j+1}{2}}\I_{B'},
\end{eqnarray} 
where to obtain the last equality we employed Eq. (\ref{FijEq25}).
Using Eq. (\ref{fijEq12}) we then obtain $\tilde{F}_{j0}$.

Thus, the remaining part of the proof of Lemma \ref{le:Fij} is to obtain the elements $F_{ij}$ such that $i,j \neq 0$ and  $i\neq j$. To do so, we use the off-diagonal elements of \eqref{FijFact2} to obtain a set of relations as follows.
\begin{obs}The following conditions hold true
\begin{equation} \label{FijFact3}
\sum_{\substack{i=1\\i\ne j}}^{d-1}\frac{1-\omega^{ki}}{1-\omega^{i}}\omega^{i/2} F_{ij}=\frac{2}{d}\omega^{\frac{j+1}{2}}\left(k+\frac{1-\omega^{kj}}{1-\omega^{j}}\omega^{j}\right), \qquad k=1,\ldots,d-1,\qquad j=1,\ldots,d-1.
\end{equation}
\end{obs}
Taking the inner product with $\bra{i} \ . \ |j\rangle$ (where $i\neq j$) on the both side of \eqref{FijObs3} we obtain
\begin{eqnarray}
 -(k-1)\omega^{ki}F_{ij}+\omega^{-\frac{1}{2}}\sum_{\substack{l=0\\l\ne i}}^{d-1}\left(\frac{\omega^{ki}-\omega^{kl}}{\omega^{i}-\omega^{l}}\right) F_{il}F_{lj}+k\omega^{(k-1)i}\omega^{-\frac{1}{2}}F_{ii}F_{ij}=\frac{\omega^{(k+1)i}-\omega^{(k+1)j}}{\omega^{i}-\omega^{j}}F_{ij}. 
\een
After rearranging some terms and using the formula for $F_{ii}$ given above, we express the above equation as
 \ben 
 \sum_{\substack{l=0\\l\ne i}}^{d-1}\frac{\omega^{ki}-\omega^{kl}}{\omega^{i}-\omega^{l}} F_{il}F_{lj}=\omega^{\frac{1}{2}}\left[\frac{\omega^{(k+1)i}-\omega^{(k+1)j}}{\omega^{i}-\omega^{j}}+\left(\frac{2k}{d}-1\right)\omega^{ki}\right]F_{ij}.
\end{eqnarray}
Next, we set $i=0$ and obtain
\begin{eqnarray}
\sum_{l=1}^{d-1}\frac{1-\omega^{kl}}{1-\omega^{l}} F_{0l}F_{lj}=\omega^{\frac{1}{2}}\left(\frac{1-\omega^{(k+1)j}}{1-\omega^{j}}+\frac{2k}{d}-1\right)F_{0j}
\end{eqnarray}
and then substitute $\tilde{F}_{0j}$ from \eqref{FijEq28} which gives
\begin{equation}\label{VihnoVerde}
    \sum_{l=1}^{d-1}\frac{1-\omega^{kl}}{1-\omega^{l}}\omega^{\frac{l+1}{2}} F_{lj}=\omega^{\frac{j}{2}+1}\left(\frac{1-\omega^{(k+1)j}}{1-\omega^{j}}+\frac{2k}{d}-1\right)\I_{B'}.
\end{equation}
Let us then take the term corresponding to $l=j$ out of the sum 
and use Eq. (\ref{TintoDeVerano}) to express $F_{jj}$. This after some simplifications
allows us to rewrite Eq. (\ref{VihnoVerde}) as
\begin{equation}
\sum_{\substack{l=1\\l\ne j}}^{d-1}\left(\frac{1-\omega^{kl}}{1-\omega^{l}}\right) \omega^{\frac{l}{2}}F_{lj}= 
\frac{2}{d}\omega^{\frac{j+1}{2}}\left(k+\frac{1-\omega^{kj}}{1-\omega^{j}}\omega^{j}\right) \I_{B'} .
\end{equation}
Finally, changing the index $l$ to $i$ leads us to \eqref{FijFact3}. Note that the equations derived for $F_{ij}$ are also valid for $\tilde{F}_{ij}$, and $\tilde{F}_{ii}=F_{ii}$.

The following observation provides the solution of \eqref{FijFact3}.
\begin{obs}
The solution of the equation \eqref{FijFact3} is
\begin{eqnarray}\label{FijEq34}
F_{ij}=-\frac{2}{d}\omega^{\frac{i+j+1}{2}}\I_{B'}, \qquad i, j=1,\ldots,d-1,\ \ \  i\ne j .
\end{eqnarray}
\end{obs} 
We multiply Eq. \eqref{FijFact3} by $\omega^{-kn}$ with $n\in\{1,\ldots,d-1\}$ such that $n\neq j$ and then 
sum both sides of the resulting formula over $k=0,\ldots,d-1$, obtaining
\begin{equation}
\sum_{\substack{i=1\\i\ne j}}^{d-1}\frac{\omega^{i/2}}{1-\omega^{i}}F_{ij}\sum_{k=0}^{d-1}\left(\omega^{-kn}-\omega^{k(i-n)}\right)=
\frac{2}{d}\omega^{\frac{j+1}{2}}\left[\sum_{k=0}^{d-1}k\omega^{-kn}+\frac{\omega^{j/2}}{1-\omega^{j}}
\sum_{k=0}^{d-1}\left(\omega^{-kn}-\omega^{k(j-n)}\right)\right]\I_{B'}.
\end{equation}
We now notice that the first sum on the left-hand side of the above 
and the last two sums on the right-hand side simply vanish due to Eq. 
(\ref{Porto}) and the fact that $n\neq j$. Exploiting then
Eq. (\ref{FijIden2}) as well as the fact that 
\begin{equation}
    \sum_{k=0}^{d-1}\omega^{k(n-i)}=d\delta_{n,i}
\end{equation}
holds true for any $x,y\in\{0,\ldots,d-1\}$, we obtain
%
%
\begin{equation}
-d\frac{\omega^{n/2}}{1-\omega^{n}}F_{nj}=\frac{2}{d}\omega^{\frac{j+1}{2}}\left(\frac{d}{\omega^{-n}-1}\right)\I_{B'},
\end{equation}
which after some manipulations leads us to (\ref{FijEq34}).
Finally, taking into account Eqs. (\ref{B2form}), \eqref{TintoDeVerano} and 
\eqref{FijEq34} we conclude that 
\be 
U B_2 U^\dagger = T_d \otimes \I_{B'}
\ee 
with $T_d$ given by Eq. (\ref{ZdTd}).
\end{proof}


\begin{fakt} The following identities hold 
\begin{eqnarray} \label{FijIden1}
  \sum_{\substack{j=0\\j\ne i}}^{d-1}\frac{1-\omega^{k(j-i)}}{1-\omega^{i-j}}=k, \qquad k=1,\ldots,d-1,\qquad i=0,\ldots,d-1,
\end{eqnarray}
and
\begin{eqnarray} \label{FijIden2} 
  \sum_{k=0}^{d-1}k\omega^{kn}=\frac{d}{\omega^n-1}, \qquad n=1,\ldots,d-1.
\end{eqnarray}
\end{fakt}
\begin{proof}
To prove the first relation (\ref{FijIden1}) we first notice that its left-hand side
can be rewritten as
\begin{equation}
      \sum_{\substack{j=0\\j\ne i}}^{d-1}\frac{1-\omega^{k(j-i)}}{1-\omega^{i-j}}=-\sum_{\substack{j=0\\j\ne i}}^{d-1}\omega^{j-i}
      \frac{1-\omega^{k(j-i)}}{1-\omega^{j-i}}.
\end{equation}
Observing then that the expression appearing on the right-hand side under the sum is 
actually a sum of a geometric series, one can rewrite the above equation as
\begin{equation}
      \sum_{\substack{j=0\\j\ne i}}^{d-1}\frac{1-\omega^{k(j-i)}}{1-\omega^{i-j}}=-\sum_{\substack{j=0\\j\ne i}}^{d-1}\left(\omega^{j-i}+\omega^{2(j-i)}+\ldots+ \omega^{k(j-i)}\right)=-\sum_{n=1}^k\omega^{-ni}\sum_{\substack{j=0\\j\ne i}}^{d-1}\omega^{nj}.
\end{equation}
Let us finally notice that for any $n=1,\ldots,d-1$ one has the following
chain of equalities
\begin{equation}\label{formulaOne}
  \sum_{\substack{j=0\\j\ne i}}^{d-1}\omega^{nj}=\sum_{j=0}^{d-1}\omega^{nj}-\omega^{ni}=-\omega^{ni}.
\end{equation}
After plugging this last formula into Eq. (\ref{formulaOne}) we arrive
at Eq. (\ref{FijIden1}). 

To show the second identity we use the following relation
\begin{eqnarray}
 \sum_{k=0}^{d-1}x^k=\frac{1-x^d}{1-x} .
\end{eqnarray}
Taking the derivative and multiplying by $x$ on both sides, we get 
\begin{eqnarray}
\sum_{k=0}^{d-1}kx^k=x\frac{\mathrm{d}}{\mathrm{d} x}\left(\frac{1-x^d}{1-x}\right)=\frac{-dx^{d}}{1-x}+\frac{x(1-x^d)}{(1-x)^2} \ .
\end{eqnarray}
Substituting $x=\omega^n$ and using the fact that $\omega^d=1$ we obtain \eqref{FijIden2}.
\end{proof}


\section{Unitary equivalence to CGLMP measurements}\label{sec:cglmp}

Let us first define the $d$-dimensional CGLMP measurements as
\begin{equation}
 A'_k=\sum^{d-1}_{r=0} \w^r \ket{r}\!\bra{r}_{A_k},\qquad
 B'_k = \sum^{d-1}_{r=0} \w^r \ket{r}\!\bra{r}_{B_k} 
\end{equation}
with $k=1,2$, where the eigenvectors are defined as
\begin{eqnarray}
 \ket{r}_{A_k}=\frac{1}{\sqrt{d}}\sum_{q=0}^{d-1} \omega^{(r-\alpha_k)q}\ket{q}, \qquad
 \ket{r}_{B_k}=\frac{1}{\sqrt{d}}\sum_{q=0}^{d-1} \omega^{-(r-\beta_k)q}\ket{q}
 \end{eqnarray}
 where $\alpha_k=(k-1/2)/2$ and $\beta_k=k/2$ \cite{BKP}.

\begin{fakt}\label{fact:cglmp}
There exist unitary operators $W_1,W_2:\mathbbm{C}^d\rightarrow \mathbbm{C}^d$ that transform $Z_d,T_d$ \eqref{ZdTd} to the $d$-dimensional CGLMP measurements in the following way: $A_1'=W_1Z_dW_1^{\dagger}, A_2'=W_1T_dW_1^{\dagger}$ and $B_1'=W_2 Z_d W_2^{\dagger},B_2'=W_2 T_d W_2^{\dagger}$.
\end{fakt}
%
%
\begin{proof}
We know that the following spectral decomposition holds, $Z_d = \sum^{d-1}_{q=0}\w^q \ket{q}\bra{q}$, and $T_d= \sum^{d-1}_{r=0} \w^r \ket{r}\!\bra{r}_{T_d} $ where
\begin{eqnarray} \label{rTd}
  \ket{r}_{T_d}= \frac{2}{d}\sum_{q=0}^{d-1}(-1)^{\delta_{q,0}}\frac{\omega^{-\frac{q}{2}}}{1-\omega^{\left(r-q-\frac{1}{2}\right)}}\ket{q} .
  \end{eqnarray}
For clarity, we verify the spectral decomposition of $T_d$,
\begin{eqnarray}\label{B20}
T_d\ket{r}_{T_d}&=&\left(\sum_{i=0}^{d-1}\omega^{i+\frac{1}{2}}\proj{i}-\frac{2}{d}\sum_{i,j=0}^{d-1}(-1)^{\delta_{i,0}+\delta_{j,0}}\omega^{\frac{i+j+1}{2}}|i\rangle\!\langle j|\right)\left(\frac{2}{d}\sum_{q=0}^{d-1}(-1)^{\delta_{q,0}} \frac{\omega^{\frac{-q}{2}}}{1-\omega^{r-q-\frac{1}{2}}}\ket{q}\right)\nonumber\\
&=&\frac{2}{d}\sum_{q=0}^{d-1}(-1)^{\delta_{q,0}}\omega^{\frac{q+1}{2}}\left(\frac{1}{1-\omega^{r-q-\frac{1}{2}}}-\frac{2}{d}\sum_{k=0}^{d-1}\frac{1}{1-\omega^{r-k-\frac{1}{2}}}\right)\ket{q}.
\end{eqnarray}
Using the relation
\be 
\sum_{l=0}^{d-1}\omega^{(r-k-\frac{1}{2})l}=\frac{2}{1-\omega^{r-k-\frac{1}{2}}},
\ee we evaluate the sum
\begin{eqnarray}
\sum_{k=0}^{d-1}\frac{1}{1-\omega^{r-k-\frac{1}{2}}}=\frac{1}{2}\sum_{l=0}^{d-1}\left(\sum_{k=0}^{d-1}\omega^{(r-k-\frac{1}{2})l}\right) .
\end{eqnarray}
The above sum is nonzero iff $l=0$, which gives
\begin{eqnarray} \label{sum10}
\sum_{k=0}^{d-1}\frac{1}{1-\omega^{r-k-\frac{1}{2}}}
=\frac{d}{2}.
\end{eqnarray}
Substituting the above relation \eqref{sum10} in Eq. $\eqref{B20}$ and the replacing the form of $\ket{r}_{T_d}$ from \eqref{rTd}, we get
\begin{eqnarray}
T_d\ket{r}_{T_d} &=& \frac{2}{d}\sum_{q=0}^{d-1}(-1)^{\delta_{q,0}}\omega^{\frac{q+1}{2}}\left(\frac{1}{1-\omega^{r-q-\frac{1}{2}}}-1\right)\ket{q} \nonumber  \\
&=&\omega^r\ket{r}_{T_d}.
\end{eqnarray}
Next we define the unitary operators $W_1,W_2$ as follows:
\begin{eqnarray} \label{Obsunist1}
W_1=M_1^\dagger FY^\dagger,\qquad W_2=SM_2^\dagger FY^\dagger,
\end{eqnarray}
where 
\begin{eqnarray}
Y=\sum_{j= 0}^{d-1}(-1)^{1-\delta_{j,0}}\omega^{\frac{d-j}{2}}\ket{j}\!\bra{j},\qquad S=\sum_{j= 0}^{d-1}\ket{j}\!\bra{d-1-j},
\end{eqnarray}
and $F$, $M_x$ $(x=1,2)$ are explicitly given by
\begin{eqnarray}
F=\frac{1}{\sqrt{d}}\sum_{i,j= 0}^{d-1}\ket{i}\!\bra{j}, \qquad M_x=\sum_{j= 0}^{d-1}\omega^{\frac{jx}{4}}\ket{j}\!\bra{j}. 
\end{eqnarray}
Here $\delta_{i,j}$ denotes the Kronecker delta, that is, $\delta_{ij}=1$ if $i=j$ and $\delta_{ij}=0$ for $i\neq j$.
To this fact it is sufficient to show that the proposed unitary transforms the eigenstates of one observable  to the eigenstates of another observable up to a complex number.
By expanding $W_1=M_1^\dagger FY^\dagger$ we obtain, 
 \begin{eqnarray}
W_1=\frac{1}{\sqrt{d}}\sum_{i,j=0}^{d-1}(-1)^{1-\delta_{j,0}}\omega^{-\frac{i}{4}+ij+\frac{j}{2}}\ket{i}\bra{j} ,
 \end{eqnarray}
 and further
 \begin{eqnarray}
 W_1^{\dagger}\ket{r}_{A_1}
 &=& \frac{1}{d}\sum_{j,q=0}^{d-1}(-1)^{1-\delta_{j,0}}\omega^{(r-j)q}\omega^{-\frac{j}{2}}\ket{j} .
 \end{eqnarray}
 The sum on the right-hand-side of the above equation would only exist if $r=j$ as $\sum_{k=0}^{d-1}\omega^{kn}=0$ whenever $n$ is a non-zero integer, and thus
 \begin{eqnarray}\label{A1toZd}
  W_1^{\dagger}\ket{r}_{A_1}
  = e^{\mathbbm{i}\pi(1-\delta_{r,0}-\frac{r}{d})} \ket{r} .
 \end{eqnarray}
Similarly, the expression
 \begin{eqnarray}
W_1^{\dagger}\ket{r}_{A_2}
 &=& \frac{1}{d}\sum_{j,q=0}^{d-1}(-1)^{1-\delta_{j,0}}\omega^{(r-j-\frac{1}{2})q}\omega^{-\frac{j}{2}}\ket{j} .
 \end{eqnarray}
 Carrying the above sum over $q$ and using the fact that $\omega^{\frac{dn}{2}}=-1$ for any non-negative odd integer $n$, with the aid of \eqref{rTd} we get
 \begin{eqnarray} \label{A2toTd}
  W_1^{\dagger}\ket{r}_{A_2}
  =-\ket{r}_{T_d} .
 \end{eqnarray}
Similarly, for Bob's observables $B'_k$,  expansion of $W_2=SM_2^\dagger FY^\dagger$ leads to
 \begin{eqnarray}
W_2=\frac{1}{\sqrt{d}}\sum_{i,j=0}^{d-1}(-1)^{1-\delta_{j,0}}\omega^{-\frac{i}{2}+ij+\frac{j}{2}}\ket{d-1-i}\bra{j} ,
 \end{eqnarray}
and further
 \begin{eqnarray}
 W_2^{\dagger}\ket{r}_{B_1}
 &=& \frac{1}{d}\sum_{j,q=0}^{d-1}(-1)^{1-\delta_{j,0}}\omega^{(j-r)q+(d-1)(\frac{1}{2}-j)-\frac{j}{2}}\ket{j} .
 \end{eqnarray}
The sum on the right-hand-side of the above equation would only exist if $r=j$ as $\sum_{k=0}^{d-1}\omega^{kn}=0$ whenever $n$ is a non-zero integer, and thus
 \begin{eqnarray}\label{B1toZd}
  W_2^{\dagger}\ket{r}_{B_1}
  = e^{\mathbbm{i}\pi(2-\frac{r-1}{d}-\delta_{r,0})} \ket{r} .
 \end{eqnarray}
 While, the following expression
 \begin{eqnarray}
 W_2^{\dagger}\ket{r}_{B_2}
 &=& \frac{1}{d}\sum_{j,q=0}^{d-1}(-1)^{1-\delta_{j,0}}\omega^{(j+\frac{1}{2}-r)q + (d-1)(\frac12-j)-\frac{j}{2}}\ket{j} ,
 \end{eqnarray}
which after carrying the sum over $q$ and using \eqref{rTd} simplifies to
 \begin{eqnarray}\label{B2toTd}
  W_2^{\dagger}\ket{r}_{B_2}
  =-\omega^{r-1}\ket{r}_{T_d}.
 \end{eqnarray}
Altogether, Eqs. \eqref{A1toZd}, \eqref{A2toTd}, \eqref{B1toZd}, and \eqref{B2toTd} complete the proof.
 \end{proof}

 \begin{fakt}\label{fact:CtoZ}
There exists a unitary operator $W_A:\mathbb{C}^d\to\mathbb{C}^d$ 
such that
\begin{equation} \label{fact3eq}
    W_A Z_dW_A^{\dagger}=a_1^{*}Z_d+2(a_1^{*})^3 T_d, \quad 
       W_A T_dW_A^{\dagger}=a_1Z_d-a_1^{*}T_d.
\end{equation} 
 \end{fakt}
 \begin{proof}
Due to the fact that the CGLMP measurements together with the maximally entangled state $\ket{\phi^+_d}$ yield the maximum violation of SATWAP Bell inequality \cite{SATWAP}, the following relations hold true,
 \begin{eqnarray}
 (B'_1)^*=a_{1}^*(A'_1)^{\dagger}+a_{1}(A'_2)^{\dagger}
 ,\qquad (B'_2)^*=a_{1}\omega^*(A'_1)^{\dagger}+a_{1}^*(A'_2)^{\dagger}.
 \end{eqnarray}
Substituting the CGLMP measurements from Fact \ref{fact:cglmp}, we get
  \begin{eqnarray}
& (W_2Z_dW_2^{\dagger})^*=a_{1}^*(W_1Z_dW_1^{\dagger})^{\dagger}+a_{1}(W_1T_dW_1^{\dagger})^{\dagger},\nonumber \\
& (W_2T_dW_2^{\dagger})^*=a_{1}\omega^*(W_1Z_dW_1^{\dagger})^{\dagger}+a_{1}^*(W_1T_dW_1^{\dagger})^{\dagger} .
 \end{eqnarray}
Taking the transpose-conjugate on both sides, and then multiplying $W_2^T$ and $W_2^*$ from left end and right end, respectively, of the above equations, we further obtain
\ben
 Z_d=a_{1}W_AZ_dW_A^{\dagger}+a_{1}^*W_AT_dW_A^{\dagger},\quad T_d=a_{1}^*\omega W_AZ_dW_A^{\dagger}+a_{1}W_AT_dW_A^{\dagger} 
\een 
where $W_A=W_2^{T} W_1$. Consequently, the above equations lead us to the desired relation \eqref{fact3eq}.
 \end{proof}


\begin{thebibliography}{11}

\bibitem{Bell}
 J. S. Bell, \href{https://doi.org/10.1103/PhysicsPhysiqueFizika.1.195}{Physics (Long Island City, N.Y.) \textbf{1}, 195 (1964).}
 
\bibitem{Bell66} J. S. Bell,
 \href{https://doi.org/10.1103/RevModPhys.38.447}{
 Rev. Mod. Phys. \textbf{38}, 447 (1966).}




\bibitem{di2}A. Ac\'in, N. Brunner, N. Gisin, S. Massar, S. Pironio, and
V. Scarani, \href{https://doi.org/10.1103/PhysRevLett.98.230501}{Phys. Rev. Lett {\bf 98}, 230501 (2007).}



\bibitem{di3}R. Colbeck, and R. Renner, \href{https://www.nature.com/articles/ncomms1416}{Nat. Commun. {\bf 2}, 411 (2011).}

\bibitem{di4}S. Pironio \textit{et al.}, \href{https://www.nature.com/articles/nature09008}{Nature {\bf 464}, 1021 (2010).}

\bibitem{di1}D. Mayers, and A. Yao, \href{https://ieeexplore.ieee.org/document/743501}{Proc. 39th Ann. Symp. on Foundations of Computer Science (FOCS), 503 (1998).}

\bibitem{NielsenChuang}M. Nielsen and I. Chuang, \href{http://mmrc.amss.cas.cn/tlb/201702/W020170224608149940643.pdf}{\textit{Quantum Computation and Quantum Information}, Cambridge University Press (Cambridge, 2000).}

\bibitem{DIEW}J.-D. Bancal \textit{et al.}, \href{https://doi.org/10.1103/PhysRevLett.106.250404}{Phys. Rev. Lett. \textbf{106}, 250404 (2011).}

\bibitem{Dimension}N. Brunner, S. Pironio, A. Ac\'in, N. Gisin, A. A. M\'ethot, and V. Scarani, \href{https://doi.org/10.1103/PhysRevLett.100.210503}{Phys. Rev. Lett. \textbf{100}, 210503 (2008).}



\bibitem{Scarani} M. McKague, T. H. Yang, and V. Scarani, \href{https://iopscience.iop.org/article/10.1088/1751-8113/45/45/455304/meta}{ J. Phys. A: Math. Theor. \textbf{45}, 455304 (2012).}

\bibitem{Yang} T. H. Yang and M. Navascu\'es, \href{https://doi.org/10.1103/PhysRevA.87.050102}{Phys. Rev. A \textbf{87}, 050102(R) (2013).}

\bibitem{Bamps} C. Bamps, and S. Pironio, \href{https://doi.org/10.1103/PhysRevA.91.052111}{Phys. Rev. A {\bf 91}, 052111 (2015).}



\bibitem{All} Y. Wang, X. Wu, and V. Scarani, \href{https://iopscience.iop.org/article/10.1088/1367-2630/18/2/025021/meta}{New J. Phys. \textbf{18}, 025021 (2016).}

\bibitem{chainedBell} I. \v{S}upi\'c, R Augusiak, A Salavrakos and A. Ac\'in, \href{https://iopscience.iop.org/article/10.1088/1367-2630/18/3/035013/meta}{New J. Phys. {\bf {18}} 035013 (2016).}




\bibitem{CHSH}
 J. F. Clauser, M. A. Horne, A. Shimony and R. A. Holt, 
 \href{https://doi.org/10.1103/PhysRevLett.23.880}{Phys. Rev. Lett. \textbf{23}, 880 (1969).}
 
 \bibitem{Acin} A. Ac\'in, S. Massar, and S. Pironio, \href{https://doi.org/10.1103/PhysRevLett.108.100402}{ Phys. Rev. Lett. \textbf{108}, 100402 (2012).}
 
\bibitem{Coladangelo}
A. Coladangelo, K. T. Goh, and V. Scarani, \href{https://www.nature.com/articles/ncomms15485}{Nat. Comm. \textbf{8}, 15485 (2017).}

\bibitem{coladangelo18} A. Coladangelo,
\href{https://doi.org/10.1103/PhysRevA.98.052115}{Phys. Rev. A \textbf{98}, 052115 (2018).}

\bibitem{Jed} J. Kaniewski, I. \v{S}upi\'c, J. Tura, 
F. Baccari, A. Salavrakos, and R. Augusiak, \textit{Maximal nonlocality from maximal entanglement and mutually unbiased bases, and self-testing of two-qutrit quantum systems}, \href{https://arxiv.org/abs/1807.03332}{arXiv:1807.03332.}


\bibitem{SATWAP} A. Salavrakos, R. Augusiak, J. Tura, P. Wittek, A. Ac\'in, and S. Pironio, \href{https://doi.org/10.1103/PhysRevLett.119.040402}{Phys. Rev. Lett. \textbf{119}, 040402 (2017).}


\bibitem{wu16} X. Wu, J.-D. Bancal, M. McKague, and V. Scarani, \href{https://doi.org/10.1103/PhysRevA.93.062121}{Phys. Rev. A \textbf{93}, 062121 (2016).}

\bibitem{mckague16} M. McKague, \href{https://iopscience.iop.org/article/10.1088/1367-2630/18/4/045013/meta}{New J. Phys. \textbf{18}, 045013 (2016). }


\bibitem{CGLMP} D. Collins, N. Gisin, N. Linden, S. Massar, and S. Popescu, \href{https://doi.org/10.1103/PhysRevLett.88.040404}{Phys. Rev. Lett. \textbf{88}, 040404 (2002). }



\bibitem{BKP}  J. Barrett, A. Kent, and S. Pironio, \href{https://doi.org/10.1103/PhysRevLett.97.170409}{Phys. Rev. Lett. \textbf{97}, 170409 (2006).}

\bibitem{Zukowski} M. \.{Z}ukowski, A. Zeilinger, and M. A. Horne
\href{https://doi.org/10.1103/PhysRevA.55.2564}{Phys. Rev. A \textbf{55}, 2564 (1997).}

\bibitem{Science}J. Wang \textit{et al.}, \href{https://science.sciencemag.org/content/360/6386/285}{Science \textbf{360}, 285 (2018).}

\bibitem{Optimal}A. Ac\'in, S. Pironio, T. V\'ertesi, and P. Wittek, \href{https://doi.org/10.1103/PhysRevA.93.040102}{Phys. Rev. A \textbf{93}, 040102(R) (2016).}


\bibitem{randomness2}E. Woodhead, J. Kaniewski, B. Bourdoncle, A. Salavrakos, J. Bowles, R. Augusiak, and A. Ac\'in, \textit{Maximal randomness from partially entangled states}, \href{https://arxiv.org/abs/1901.06912}{arXiv:1901.06912.}


\bibitem{andrea}A. Coladangelo, \textit{A two-player dimension witness based on embezzlement, and an elementary proof of the non-closure of the set of quantum correlations}, \href{https://arxiv.org/abs/1904.02350}{arXiv:1904.02350.}






\bibitem{Delegated}B. W. Reichardt, F. Unger, and U. Vazirani, \href{https://www.nature.com/articles/nature12035}{Nature \textbf{496}, 456 (2013).}

\bibitem{related}I. \v{S}upi\'c, D. Cavalcanti, J. Bowles, \textit{Device-independent certification of tensor products of quantum states using single copy self-testing protocols}, \href{https://arxiv.org/abs/1909.12759}{arXiv:1909.12759}







\end{thebibliography}
\end{document}